\documentclass[12pt]{amsart}

\usepackage[english]{babel}
\usepackage{amsmath,amssymb,amsthm}
\usepackage{xcolor}

\usepackage[colorlinks=true,citecolor=blue,linkcolor=blue,urlcolor=orange]{hyperref}
\usepackage[round]{natbib}
\usepackage{enumitem}
\usepackage{tikz}
\usetikzlibrary{arrows}

\usepackage{geometry}

\usepackage{wrapfig}

\usepackage{mathtools}

\usepackage{todonotes}
\usepackage{xfrac}

\usepackage[onehalfspacing]{setspace}

\usepackage{thmtools}

\usepackage{kantlipsum} 

\usepackage{float}

\usepackage{soul}
\usetikzlibrary{shapes,backgrounds}

\theoremstyle{plain} 
\newtheorem{theorem}{Theorem}[section]
\newtheorem{proposition}[theorem]{Proposition}
\newtheorem{lemma}[theorem]{Lemma}
\newtheorem{corollary}[theorem]{Corollary}

\theoremstyle{definition}
\newtheorem{definition}[theorem]{Definition}
\newtheorem*{observation}{Observation}

\newtheorem{innercustomthm}{Example}

\newenvironment{customthm}[1]
{\renewcommand\theinnercustomthm{#1}\innercustomthm}
{\endinnercustomthm}

\let\Oldenddefinition\enddefinition
\def\enddefinition{\hfill $\triangleleft$\Oldenddefinition}%

\begin{document}

\title[]{Irrational random utility models$^*$} 
\thanks{$\dagger$ WZB Berlin. E-mail: daniele.caliari@wzb.eu}
\thanks{$\ddagger$ University of Gothenburg. E-mail: henrik@petri.se}
\thanks{$^*$ We would like to thank the following for helpful comments and suggestions: Dilip Ravindran, Victor Aguiar,
Jose Apesteguia, Miguel Ballester, Christian Basteck, Valentino Dardanoni, Tilman Fries, Gilat Levy, Paola Manzini, Marco Mariotti, Yusufcan Masatlioglu, Ran Spiegler, Ran Shorrer, J\"{o}rg Stoye, and Ferdinand Vieider; the participants of the 10th BRIC conference 2024, and internal seminars at the WZB and Humboldt University.}

\author[]{Daniele Caliari$^{\dagger}$  \hspace{0.2em} \& \hspace{0.05em} Henrik Petri$^{\ddagger}$ }

\maketitle

\begin{abstract}

The Random Utility Model (RUM) is the leading model to represent the aggregate choices of a heterogeneous population of preference maximizers. We show that if (and only if) preferences are sufficiently uncorrelated, RUM choices can also be generated by a population of decision makers who do not maximize any preference. In proving this result, we also characterize the general class of choices generated by such irrational populations, with applications beyond the RUM framework. We discuss the relevance of our results for the falsifiability of the rational interpretation of RUMs, the inference of individual rationality from aggregate choices, and the nature of welfare judgments.

\end{abstract}
\hspace*{3,6mm}{\small \textbf{Keywords}: Stochastic choice, random utility models, rationality.} 

\hspace*{3,6mm}{\small \textbf{JEL Classification}: D00, D90, D91.}

\newpage

\section{Introduction}

The random utility model (RUM) is the most renowned model of stochastic choice and it is the leading notion of stochastic rationality in economics (\cite{mcfadden1990stochastic}, \cite{mcfadden2006revealed}): choice probabilities satisfy the RUM hypothesis if they result from the aggregation of the choices of rational - preference maximizing - decision-makers. These aggregate choices are normally referred to as \emph{stochastically rational}.

A starting point of the present paper is the observation that aggregate choices may be stochastically rational even if all decision-makers are irrational, i.e. they violate the tenets of rationality.\footnote{In this paper, we assume that a “rational” decision-maker is one who maximizes a strict linear order.  This notion is equivalent to Sen's property $\alpha$ \citep{sen71} or the weak axiom of revealed preference (WARP) \citep{arrow1959rational}. We refer to section \ref{sec: preliminaries} for precise definitions.}
 The intuition behind this observation is simple. If the irrationalities of decision-makers are sufficiently uncorrelated, they will cancel out, and aggregate data will appear \emph{as if} stochastically rational.\footnote{This intuition is not new. \cite{becker} famously pointed to it while \cite{grandmont1992transformations} showed that, in the context of demand functions, heterogeneity helps to yield a nicely behaved aggregate demand.} We call the subset of RUMs that display this property "\emph{irrational}" [I-RUM].

To capture this intuition, we introduce novel bounds on preference heterogeneity that we call \emph{correlation bounds}. The use of the term "correlation" comprises both their general role in governing the relationship between individual and aggregate choices (joint and marginal distributions), and the specific application in our paper as upper bounds on the correlation between the preferences of the rational decision-makers.\footnote{Our correlation bounds are a strengthening of the lower Fr\'echet inequalities, which are bounds on the conjunction of events; in our case, individual choices from different menus of alternatives.} In a nutshell, a population of rational decision-makers is, in aggregate, observationally equivalent to a population of irrational ones if, and only if, their preferences are sufficiently uncorrelated.

The paper begins with the characterization of irrational representations, which, in seemingly sharp contrast to RUMs, we call \textit{random non-utility models} (RNUMs). This result is interesting per s\'e as it opens the door to irrational representation theorems in which aggregate choices can be rationalized by some well-known class of (behavioral) choice functions. Then, in our main theorem, we overlap rational (RUM) and irrational (RNUM) representations to characterize the set of I-RUMs using the correlation bounds. Interestingly, the family of I-RUMs turns out to be the convex hull of the dual RUMs that satisfy the correlation bounds, an influential and well-studied subset of RUMs \citep{mariottidual}. We illustrate the intuition behind our main result in one simple example.

\begin{innercustomthm} \label{ex1.1}

An analyst observes the aggregate choice probabilities in all non-empty menus of a set of three alternatives - "chicken" ($c$), "steak" ($s$), and "frogs' legs" ($f$) - from a population where 1/2 of the decision-makers have preference $c \hspace{0.15em} P_1 \hspace{0.15em} s \hspace{0.15em} P_1 \hspace{0.15em} f$ and 1/2 have preference $s \hspace{0.15em} P_2 \hspace{0.15em} c \hspace{0.15em} P_2 \hspace{0.15em} f$. The first and second column in the table below displays the choices of these rational decision-makers from each menu with $P_1$'s choices highlighted in \textcolor{red}{red} and $P_2$'s in \textcolor{blue}{blue}. 

\begin{table}[H]
\centering 
\begin{tabular}{c | c c | c c} 
\hline\hline 
 Menu & $\textcolor{black}{c_{P_1}(A)}$ & $c_{P_2}(A)$ & $c_1(A)$  & $c_2(A)$  \\ [0.1ex] 
 \hline 
$csf$ & $\textcolor{red}{c}$ & $\textcolor{blue}{s}$ & $\textcolor{blue}{s}$ & $\textcolor{red}{c}$\\  
$cs$ & $\textcolor{red}{c}$  & $\textcolor{blue}{s}$ & $\textcolor{red}{c}$  & $\textcolor{blue}{s}$ \\ 
$cf$ & $\textcolor{red}{c}$ & $\textcolor{blue}{c}$ & $\textcolor{red}{c}$   & $\textcolor{blue}{c}$  \\ 
$sf$ & $\textcolor{red}{s}$ & $\textcolor{blue}{s}$ & $\textcolor{red}{s}$  & $\textcolor{blue}{s}$  \\ 
\hline\hline 
\end{tabular}
\label{table: portfolio} 
\end{table} 

We now provide an interpretation of the aggregate data from an irrational perspective. Consider the Luce \& Raiffa's dinner example \citep{luce1957introduction} in which a decision-maker chooses chicken when only steak is available, but switches to steak whenever also frogs' legs are available even if he dislikes frogs' legs. Imagine now another decision-maker with the same behavior but who chooses steak when only chicken is available and switches to chicken in the presence of frogs' legs. These well-known irrational behaviors are summarized by the choices $c_1$ and $c_2$ in the table. A population in which 1/2 of the decision-makers choose according to $c_1$ and $c_2$ will induce the same aggregate data as the original population of rational decision-makers.
\end{innercustomthm}

Example \ref{ex1.1} clarifies the conditions for the existence of irrational representations, as if the probability distribution on the preferences $P_1$ and $P_2$ is perturbed only slightly, an irrational representation no longer exists. To see this, imagine an analyst who observes the aggregate choices of a (rational) population where 51\% of decision-makers have preference $P_1$ ($c \hspace{0.15em} P_1 \hspace{0.15em} s \hspace{0.15em} P_1 \hspace{0.15em} f$) and 49\% have preference $P_2$ ($s \hspace{0.15em} P_2 \hspace{0.15em} c \hspace{0.15em} P_2 \hspace{0.15em} f$). Knowing only the aggregate choices, she is oblivious to the individual rationality of the decision-makers, and so she searches for an irrational interpretation similar to the one in Example \ref{ex1.1}. We next argue that an irrational representation does not exist. Indeed, at least $2\%$ of the decision-makers \emph{must be} rational with preference $P_1$. Note first that the only choice functions consistent with aggregate choices are those in the preceding table (i.e. $c_{P_1},c_{P_2},c_1,c_2$). This is because $f$ is never chosen in any menu.  Suppose now that \emph{less} than $2\%$ of the decision-makers are rational with preference $P_1$. Since $51\%$ of the decision-makers choose $c$ in $csf$, there must be more than $51-2=49\%$ irrational decision-makers who choose $c$ in $csf$. But, any irrational decision-maker who chooses $c$ in $csf$ must choose $s$ in $cs$ (to violate rationality), i.e. use choice function $c_2$. Hence, more than $49\%$ of decision-makers must choose $s$ in $cs$, which is incompatible with an aggregate share of $49\%$ choosing $s$ in $cs$.

By considering irrational decision-makers, our main result shows that the common interpretation of stochastically rational data is sometimes undermined. This observation questions the falsifiability of the rational interpretation of the RUM hypothesis, i.e. even if all decision-makers violate rationality the RUM hypothesis may fail to be rejected,\footnote{The problem of falsifiability of the rationality hypothesis was noticed, within a theory of demand framework, by \cite{blundell2003nonparametric}: "revealed preference tests are unlikely to reject the integrability conditions [\emph{aggregate rationality}] for such data but it is not clear that we would wish to characterize them as the outcome of a 'rational' procedure."} a problem we discuss in the second part of the paper. First, we introduce a class of representations in which only a fraction $\sigma$ of decision-makers are irrational ($\sigma$-RNUMs). Leveraging this result, we show that almost all RUMs can be represented by a population where a strictly positive fraction $\sigma$ of decision-makers are irrational, and in fact, that all RUMs with full support or that satisfy positivity (e.g., Logit model) have this property.\footnote{A RUM with full support has strictly positive probability mass on every preference.} We then generalize the non-falsifiability problem and show that, under the RUM hypothesis, there is (almost) always a fraction of decision-makers whose behavior is unconstrained, i.e. these decision-makers' rationality can be violated in the most extreme sense. This result is deliberately demanding from an irrationality perspective, implying that such a fraction of irrational decision-makers may be small. We, therefore, go on to show that the fraction of irrational decision-makers that are compatible with the RUM hypothesis is larger the less constrained the set of possible irrational behaviors is. These results have important consequences for empirical applications where irrational behavior is unlikely to be extreme (i.e., more constrained), implying that the non-falsifiability issue is likely to be widespread.

Finally, we conclude by discussing the inference of individual rationality from aggregate data. This issue is prominent in the literature \citep{blundell2003nonparametric} because welfare judgments are founded on the rational interpretation of the RUM hypothesis. In Example \ref{ex1.1}, for instance, if the analyst's prior knowledge on individual rationality is sufficient, the aggregate choices would imply that 1/2 of the decision-makers have preference $c \hspace{0.15em} P_1 \hspace{0.15em} s \hspace{0.15em} P_1 \hspace{0.15em} f$ and 1/2 have preference $s \hspace{0.15em} P_2 \hspace{0.15em} c \hspace{0.15em} P_2 \hspace{0.15em} f$. On the other hand, insufficient information about individual rationality may lead the analyst to interpret the data from an irrational perspective concluding that the decision-makers have menu-dependent preferences for which welfare judgments are uncertain. We follow the footsteps of \cite{hoderlein2014revealed} to provide lower bounds for the fraction of rational decision-makers.\footnote{We refer the reader to \cite{hoderlein2014revealed} and \cite{hoderlein2015testing} for a more detailed discussion on lower/upper bounds of individual rationality when only aggregate data are observed.} This derives directly from the characterization of $\sigma$-RNUMs and nicely contrasts classical characterization results of RUMs (\cite{falmagne1978representation}, \cite{barbera1986falmagne}) that focus on the upper bound, where all decision-makers are rational. Our final discussion aims to show that individual rationality is conceptually unrelated to the RUM hypothesis - a point that we illustrate with some provoking examples in both an abstract and an applied setting \citep{kitamura2018nonparametric}.

\subsection{Related Literature}

Our study of irrational behavior within the RUM contributes to a longstanding literature on stochastic choice. Most closely, our framework aligns with recent papers that also study probability distributions on choice functions and their aggregation. \cite{DMMT} and \cite{dardanoni2023mixture} refer to this approach as "mixture choice functions" while \cite{filiz2023progressive} as "random choice models". In this paper, we will adopt this second nomenclature. The focus of these papers is, however, different from ours as they tackle the well-known identification issues of stochastic choice models. \cite{DMMT} provide conditions under which an analyst who only observes aggregate choice probabilities can identify cognitive parameters under specific models and restrictive assumptions (preference homogeneity or known distribution of preferences). \cite{dardanoni2023mixture} refine these results by assuming that the analyst observes each decision-maker's choices and show that this allows the identification of both cognitive parameters and preferences. \cite{filiz2023progressive}, instead, study the aggregation of a population of decision-makers who are possibly boundedly rational under the assumption that the collection of their choices is progressive - a generalization to the well-known single-crossing property \citep{apesteguia2017single}. The authors show that this restriction on the set of choice functions allows the identification of the unique sources of heterogeneity in the data.\footnote{See also \cite{petri23} for a discussion of these issues (in a setting of multivalued choice).} In section \ref{sec: extension}, we discuss some connections between our results and the models of \cite{apesteguia2017single} and \cite{filiz2023progressive}. In a recent contribution, \cite{chambers2025ordered} explore the relation between identification properties of random choice models and the Fr\'echet bounds. This paper nicely complements our analysis for two reasons: (i) its focus on identification instead of falsifiability; and (ii) on rational (or minimally irrational) decision-makers. 

We also contribute to the longstanding literature on the characterizations of the RUM (\cite{block}, \cite{falmagne1978representation}, \cite{barbera1986falmagne}, \cite{mcfadden1990stochastic}, \cite{gilboa1990necessary}, \cite{fiorini2004short}, \cite{mcfadden2006revealed}, \cite{stoye2019revealed}), and its special cases such as single-crossing RUM \citep{apesteguia2017single}, and dual RUM \citep{mariottidual}; as well as the study of its non-identifiability issues (\cite{turansick}, \cite{suleymanov2024branching}). More specifically, our work is closely related to \cite{mariottidual} as dual RUMs will be crucial in our characterization of I-RUMs, and to \cite{turansick}, \cite{suleymanov2024branching} as by considering irrational decision-makers our results extend further the non-identifiability issues of the RUM.

Identifiability and falsifiability under the RUM hypothesis are tightly connected issues as they both leverage the tension between aggregate and individual behavior, an old problem within the literature on the theory of demand. \cite{becker} first recognized that the aggregation of erratic consumers may lead to a well-behaved aggregate demand function. The connection between heterogeneity and aggregate behavior was then investigated more in-depth by \cite{grandmont1992transformations} who, without relying on individual rationality, provided sufficient conditions for the distribution of individual demand functions that guarantee a well-behaved aggregation. These issues have been long recognized within the literature of demand estimation that often incorporates tests of individual rationality, e.g. see \cite{blundell2003nonparametric}. In this tradition, two relevant papers for us are \cite{hoderlein2014revealed}, and \cite{hoderlein2015testing}. The authors ask broadly what knowledge of the "joint distribution of demand", i.e. individual demand functions, the analyst may have by knowing only the marginal distribution of demand, i.e. aggregate choice probabilities. Their setting is therefore similar to ours and, to the best of our knowledge, they are the first to introduce Fr\'echet bounds to bound the fraction of (ir)rational consumers, i.e. they prove, within the context of the theory of demand, the necessity of these bounds. Our main result abstracts from the framework of the theory of demand to provide necessary and sufficient conditions for stochastic rationality to coexist with individual irrationality. We believe our abstraction could lead to the application of these intuitions to broader contexts where different definitions of (ir)rationality may apply.\footnote{For example, \cite{cattaneo2020random} introduce Random Attention Models (RAMs) and show that they can be characterized as RCMs with support on the (limited attention) deterministic choice functions characterized in \cite{masa2012}. A potential avenue of research is whether stochastic choice functions that are RAMs can be represented as RCMs with support outside the choice functions with limited attention, or equivalently, if there is an interpretation other than attention to the aggregate choices. In this sense, recently \cite{kashaev2022random} generalized RAMs to allow for preference heterogeneity, which we show to be crucial for the existence of irrational representations. More generally, this problem can be phrased as: "Under which conditions can one represent the aggregate choices from a model with the negation of the model?" We thank Ran Spiegler for this general intuition behind our result.} 

The view of our results from the lens of the falsifiability of the rational interpretation of the RUM hypothesis is, instead, relatively novel. \cite{im2022non} focus on a restrictive setting of two consumption goods and two observations to show that unless more than half of the population is irrational, the population could be stochastically rationalizable overall. They refer to this notion as "false acceptance of stochastic rationalizability". In this paper, we generalize their observation by characterizing the exact conditions under which stochastic rationalizable choices have an irrational representation. In our discussion section, we borrow some of their observations to show, within the framework of the theory of demand, that considering only rational representations of RUMs may lead to welfare misjudgments. 

Finally, our results contribute to the literature on tests of the RUM hypothesis: \cite{kitamura2018nonparametric}, \cite{mccausland2020testing}, \cite{aguiar2021stochastic}, \cite{smeulders2021nonparametric}, \cite{deb2023revealed}. More specifically, we discuss our results in the framework of \cite{kitamura2018nonparametric}. The authors provide a statistical test for the RUM hypothesis based on the distance between the observed choice probabilities and the set of RUMs which implies that whenever the distance is zero, i.e. the choice probabilities are stochastically rational, the RUM hypothesis is accepted. This observation is a direct consequence of the restriction of the set of choice functions to the rational ones. In weakening this assumption, our results show that the sole focus on aggregate choice probabilities often does not provide any evidence about individual behavior and that if evidence of individual rationality exists they often do not depend on whether aggregate choice probabilities are rational or not.

\section{Preliminaries}\label{sec: preliminaries}

We denote by $X$ a finite set of alternatives with $\vert X \vert = N$.  A subset $A \subseteq X$ is called a \emph{menu}.\footnote{With a little abuse of notation, we use the multiplicative notation $ab$, $abc$ for menus $\{a,b\}, \{a,b,c\}$.}  Let $\mathcal{A} \subseteq 2^{X} \setminus \{\emptyset\}$ denote the collection of nonempty menus from which choices are observed. We assume that $\mathcal{A}$ contains all menus of cardinality $2$ and $3$ (referred to as binary and ternary menus throughout the paper). Let $K$ denote the cardinality of the collection of menus, i.e. $| \mathcal{A} | = K$. 

The empirical primitive is a \emph{stochastic choice function}, i.e. a map $\rho: X \times \mathcal{A} \to [0,1]$ such that i) $\sum_{a \in A} \rho(a, A) = 1$ for all $A \in \mathcal{A}$ and ii) $\rho(a, A)=0$ for all $a \in X \setminus A$. 
 
Our decision-makers choose once from several menus. Their choices are described by a \emph{choice function}, i.e. a map $c: \mathcal{A} \to X$ such that $c(A) \in A$ for all $A \in \mathcal{A}$. We are particularly interested in choices that result from the act of preference maximization. Let $P$ be a \emph{strict linear order}\footnote{A binary relation $P$ is a strict linear order, if it is \emph{asymmetric} (if $aPb$ then not $bPa$), \emph{transitive} ($aPb$ and $bPc$ implies $aPc$) and \emph{complete} ($a \neq b$ implies $a P b$ or $b P a$).}, henceforth "preference", we denote by $c_P$ a choice function that is rationalizable by a preference and refer to it as \emph{rational}, i.e. $c_P(A) = \max(P, A) = \{a \in A: aPb \,\, \text{for all} \,\, b \in A \setminus \{a\} \}$ for all $A \in \mathcal{A}$. A choice function $c$ satisfies \emph{property} $\alpha$  if $A, B \in \mathcal{A}$, $a \in B \subseteq A$ and $a = c(A)$ then $a = c(B)$. \cite{sen71} shows that a choice function is rationalizable if and only if it satisfies property $\alpha$.\footnote{In Appendix \ref{app: online}, we have generalized our results to incomplete domains.} We will use this equivalence throughout the paper. If a choice function is not rationalizable we will refer to it as \emph{irrational}. Thus, a choice function is irrational if and only if it violates property $\alpha$. We let $\mathcal{C}$ denote the collection of all choice functions and $\mathcal{P}$ the set of rationalizable choice functions. Note that there is a one-to-one correspondence between the set of rational choice functions and the set of preferences, therefore, throughout the paper we refer to them as $c_P$ or $P$ interchangeably. Finally, and most importantly, we do not assume that the analyst observes the individual choice functions even if, especially in experimental settings, this is often the case.

\subsection{Representations}

A \emph{random choice model (RCM)} is a probability distribution $\mu$ on $\mathcal{C}$.\footnote{We borrow the term "Random Choice Model" (RCM) from \cite{filiz2023progressive}.} The support of $\mu$, denoted $\text{supp}(\mu)$, is the set of choice functions with strictly positive probability, i.e. $\{c \in \mathcal{C}: \mu(c) > 0\}$. The RCM stochastic choice function is $\rho_{\mu}(a,A)=\mu(\mathcal{C}(a,A))$ for all $a \in A$ and $A \in \mathcal{A}$, where $\mathcal{C}(a,A) = \{c \in \mathcal{C}: a = c(A)\}$. 

\begin{definition}
\leavevmode
\begin{enumerate}
    \item A \emph{random utility model} (RUM) is an RCM with support in $\mathcal{P}$.
    \item A \emph{random non-utility model} (RNUM) is an RCM with support in $\mathcal{C} \setminus \mathcal{P}$.
\end{enumerate}
\end{definition}

If $\rho$ is a stochastic choice function such that $\rho=\rho_{\mu}$ for some RUM (resp. RNUM) $\mu$, we will (with some abuse of terminology) call $\rho$ itself a RUM (resp. RNUM). Finally, we say that an RCM is a \emph{dual RCM} if the cardinality of its support is less than or equal to two (i.e. if $|\text{supp}(\mu)| \leq 2$). An important subset of dual RCMs is that of \emph{dual RUMs}, whose support consists of rational choice functions. Finally, we say that a RUM has full-support if $\text{supp}(\mu) = \mathcal{P}$.\footnote{Note that the standard definition of RUM as a probability distribution on preferences \citep{block} is here substituted by a probability distribution on rational choice functions (\cite{mcfadden1990stochastic}, \cite{kitamura2018nonparametric}, \cite{stoye2019revealed}). Given the one-to-one correspondence between the two sets the two definitions are equivalent from the perspective of the aggregate choice probabilities. However, they are conceptually different. The requirement $\text{supp}(\mu) \subseteq \mathcal{P}$ implies that each decision-maker is rational throughout all choices, while as shown in Example \ref{ex1.1}, the same $\rho$ may also be represented by irrational decision-makers.}

\section{Random Non-Utility Models} \label{sec: RNUM} 

As mentioned in the introduction, we first provide a characterization of RNUMs. Only later, in section \ref{sec: main}, we apply this result to study irrational representations of RUMs. This is for two reasons. First, the general case of RNUMs is interesting per s\'e as it provides a framework to analyze (behavioral) irrational representations from a general viewpoint. Second, although the latter result is a corollary (being the intersection of RUMs and RNUMs), there are some important differences between the two results that we highlight in section \ref{sec: main}.

A key observation in our characterization of RNUMs, is that the relationship between $\rho$ and $\mu$ is analogous to the one between joint and marginal distributions and, as a result, it is governed by the Fr\'echet bounds. In a nutshell, these bounds determine the set of joint distributions compatible with given marginals (see Appendix \ref{app: frechet} for a discussion on the Fr\'echet bounds and their application to our framework). We will apply a variation of the Fr\'echet bounds that we call "correlation bounds". This is to highlight the role of the correlation between rational and irrational choice functions that will be soon apparent (see lemma \ref{lemma: CRNUM}). In Appendix \ref{app: frechet}, we also motivate the use of the word "correlation" with a more mathematical intuition. We develop the intuition behind the correlation bounds within our opening example.

\begin{customthm}{1}\label{ex1.1rnum} 
We would like to understand when an SCF $\rho$ has an RNUM representation. Consider the following aggregate data from our opening example where $0.5 + \frac{x}{2}$ have preference $c P_1 s P_1 f$ and $0.5 - \frac{x}{2}$ have preference $s P_2 c P_2 f$. 

\begin{table}[H]
\centering 
\begin{tabular}{c | c c c } 
\hline\hline 
 Menu & $\textcolor{black}{\rho(c,\cdot)}$ & $\rho(s,\cdot)$ & $\rho(f,\cdot)$ \\ [0.1ex] 
 \hline 
$csf$ & $0.5 + \frac{x}{2}$ & $0.5 - \frac{x}{2}$ & $0$ \\  
$cs$ & $0.5 +  \frac{x}{2}$  & $0.5 -  \frac{x}{2}$ & $0$ \\ 
$cf$ & $1$ & $0$ & $0$  \\ 
$sf$ & $0$ & $1$ & $0$ \\ 
\hline\hline 
\end{tabular}
\end{table}

First, note that any RCM representation $\mu$ of this data must be such that the fraction of rational decision-makers is at least $x$, more specifically, with preference $cP_1sP_1f$. To see this, note that a fraction $\rho(c,cs)$ choose $c$ in $cs$ and a fraction $\rho(c,csf)$ choose $c$  in $csf$. Since $\rho(c,cs) + \rho(c,csf)= 0.5+\frac{x}{2} + 0.5 +\frac{x}{2} = 1 + x$, this means that there is an overlap $x$ of the decision-makers that choose $c$ in $cs$ and those that choose $c$ in $csf$.\footnote{We here apply the simplest case of the (lower) Fréchet bounds, which for any two events $A,B$ say that the probability $P(A \cap B)$  is (weakly) larger than $P(A)+ P(B) -1 $. Intuitively, if $P(A) + P(B)$ is larger than $1$ then there will be (at least) a $P(A) + P(B) -1$ "overlap" between the events $A$ and $B$.}  I.e. at least a fraction $x$ choose $c$ in $cs$ \emph{and} $c$ in $csf$. But, since everyone chooses $c$ in $cf$ and $s$ in $sf$, this in turn implies that at least a fraction $x$ has preference $P_1$. Thus, if $x \neq 0$, the data is incompatible with an RNUM representation.

\end{customthm}

In the remainder of the section, we show that this reasoning can be generalized to every stochastic choice function and characterize the set of RNUMs. But first, we introduce a few more pieces of notation to simplify the exposition. Denote by $x_{P}(N-1), x_{P}(N)$ the worst two alternatives according to a preference $P$, and for each preference $P$ let $\mathcal{A}(P)= \mathcal{A} \setminus \{\{x_{P}(N-1),x_{P}(N)\}\}$.\footnote{The exclusion of this menu can be explained from the perspective of rationality, i.e. Sen's property $\alpha$: since neither $x_{P}(N-1)$ nor $x_{P}(N)$ is ever chosen by $c_P$ from any other menu, the menu $\{x_{P}(N-1),x_{P}(N)\}$ can never be involved in a violation of property $\alpha$.} We next define for each preference $P$:

$$\mathbb{C}^{\rho}_{P} = \sum_{A \in \mathcal{A}(P)} \rho(c_{P}(A), A).$$

The quantity $\mathbb{C}^{\rho}_P$ represents the probability mass in $\rho$ that "correlates" with a rational choice function $c_P$. In other words, it measures positively the correlation between the choices of the decision-makers in the population w.r.t. a specific preference. Given our opening example, the correlation bounds come as upper bounds on the correlation between the choices of the decision-makers. The correlation bounds require the left-hand side of the Fr\'echet inequality in Lemma \ref{lemma: strict1} (applied to $c_P$ and $\mathcal{A}(P)$) to be less than or equal to zero, noting that $|\mathcal{A}(P)| = K-1$.

\begin{definition}[Correlation Bounds] \label{def: CB}
A stochastic choice function $\rho$ satisfies the \emph{Correlation Bounds} if for all preferences $P$ it holds that
$$\mathbb{C}^{\rho}_{P} - [K-2] \leq 0.$$\end{definition}

Throughout, we say that $\rho$ satisfies the correlation bounds \emph{with equality} if it satisfies the correlation bounds and $\mathbb{C}^{\rho}_{P} - [K-2] = 0$ for some preference $P$ (in which case we also say that the bounds are satisfied with equality \emph{at} $P$), and that it satisfies them \emph{with strict inequality} if $\mathbb{C}^{\rho}_{P} - [K-2] < 0$ for all preferences $P$. We use the same terminology for an RCM $\mu$ whenever $\rho_{\mu}$ has the corresponding property.

Importantly, the bounds can be rewritten as a function of the representation, instead of the stochastic choice function, to highlight the correlation between rational and irrational choice functions.

\begin{definition}[Measure of correlation between preference and choice functions] Fix a preference $P$, then for all choice functions $c$:
$$n(P,c) = |\{A \in \mathcal{A}(P): c_P(A) = c(A)\}|.$$
\end{definition}
The following lemma qualifies our use of the word "correlation", and also shows that, given a preference $P$, the quantity $\mathbb{C}^{\rho}_{P}$ is the same regardless of which RCM $\mu$ represents $\rho$.
\begin{lemma}[$\mathbb{C}^{\rho}_{P}$ for any RCM] \label{lemma: CRNUM}
Let $\rho$ be a stochastic choice function and let $\mu$ be an RCM s.t. $\rho=\rho_{\mu}$. For all preferences $P$ it holds that
$$\mathbb{C}^{\rho}_{P} = \sum\limits_{c \in \mathcal{C}} \mu(c) n(P,c).$$
\end{lemma}
\begin{proof} Let $\mu$ be s.t. $\rho=\rho_{\mu}$ and let $P$ be a preference. The proof follows by noting that $$ \sum_{A \in \mathcal{A}(P)} \rho(c_P(A),A) = \sum_{A \in \mathcal{A}(P) }\sum_{c \in \mathcal{C}}[\mu(c)\bold{1}\{c_P(A) = c(A) \} ]=$$
$$\sum_{c \in \mathcal{C}} \mu(c) \sum_{A \in \mathcal{A}(P)} \bold{1}\{c_P(A) = c(A) \} = $$
$$  \sum_{c \in \mathcal{C}} \mu(c)n(P, c).$$
The first equality follows since $\rho=\rho_{\mu}$, the second equality follows by changing the order of summation, and the final one follows by definition of $n(P, c)$. \end{proof}

We are now ready to state our first main result.

\begin{theorem} \label{thm: RNUM}
A stochastic choice function $\rho$ is an RNUM if and only if it satisfies the correlation bounds.
\end{theorem}

\noindent \emph{Proof (sketch)}. We delegate the full proof to Appendix \ref{app: main}. The necessity of the correlation bounds follows since they are defined as variations of (lower) Fr\'echet bounds.\footnote{Please see Appendix \ref{app: frechet} for a discussion of the Fréchet bounds and Appendix \ref{sec: nproof} for a proof of necessity.} We next provide a sketch of sufficiency. First, because the correlation bounds are linear inequalities, the set of stochastic choice functions that satisfy them is convex. By the Krein--Milman theorem in finite dimension (Lemma \ref{lemma: cat}), any point of this set can be written as a convex combination of its extreme points; therefore, if we can show that each extreme point of this set has an RNUM representation, the proof is complete, because a mixture of RNUMs is itself an RNUM. To show this, we proceed through three lemmas. By Lemma \ref{lemma: extreme}, any extreme point either is an RNUM, and we are done, or satisfies the correlation bounds with equality at some preference $P$. By Lemma \ref{lemma: correq}, any such point can be written as a convex combination of dual RCMs that satisfy the correlation bounds with equality at the same $P$. Finally, by Lemma \ref{lemma: dual}, each such dual RCM has an RNUM representation.

\section{Irrational random utility models} \label{sec: main}

Recall, a stochastic choice function is a RUM if a probability distribution on a set of rational choice functions can describe it. However, even if aggregate choices are stochastically rational, they may hide a population of completely irrational decision-makers, questioning the rational foundations of the RUM. This motivates us to apply Theorem \ref{thm: RNUM} to study RUMs that result from the choices of a population of irrational decision-makers.

\begin{definition}A stochastic choice function $\rho$ is an \emph{Irrational RUM (I-RUM)} if it is both a RUM and an RNUM.\end{definition}

Applying Theorem \ref{thm: RNUM} to RUMs provides a novel viewpoint on the relation between heterogeneous preferences and irrational choice functions. Specifically, the characterization of I-RUMs is based on the intuition that only a RUM with a sufficiently spread-out probability distribution can be an I-RUM. We develop this intuition within our opening example. 

\begin{customthm}{1}[continued] \label{ex1.1cc}
In Example \ref{ex1.1}, we have shown that a RUM $\mu$ with support on $P_1$ and $P_2$ has an irrational representation (RNUM) only if $\mu(P_1) = \mu(P_2) = \frac12$. To formalize this result using the correlation bounds, note that an irrational decision-maker who chooses alternative $c$ from the grand set $csf$ must either choose $s$ from $cs$ or $f$ from $cf$ (to violate Sen's property $\alpha$). As a result, an I-RUM exists only if $\rho(s, cs) + \rho(f, cf) \geq \rho(c, csf)$. Since $f$ is never chosen, this observation is sufficient to show that a RUM with support $P_1$ and $P_2$ is an I-RUM only if $\mu(P_2)= \rho(s,cs)= \rho(s, cs) + \rho(f, cf)  \geq \rho(c,csf) = \mu(P_1)$. Symmetric reasoning implies that $\mu(P_1) = \mu(P_2) = 1/2$.

\begin{table}[H]
\centering 
\begin{tabular}{c | c c | c } 
\hline\hline 
 Menu & $\textcolor{black}{\max(A,P_1)}$ & $\max(A,P_2)$ & $\max(A,P_3)$ \\ [0.1ex] 
 \hline 
$csf$ & $c$ & $s$ & $f$ \\  
$cs$ & $c$  & $s$ & $s$ \\ 
$cf$ & $c$ & $c$ & $f$  \\ 
$sf$ & $s$ & $s$ & $f$ \\ 
\hline\hline 
\end{tabular}
\end{table} 

To better understand the mechanism underlying our characterization of I-RUMs, we next look at RUMs with support on preferences $P_1$ and $P_3$. Applying the same reasoning as above, one can see that $\mu$ is an I-RUM only if $\mu(P_3)+ \mu(P_3) = \rho(s, cs) + \rho(f, cf)  \geq \rho(c,csf) = \mu(P_1)$ implying that $\mu(P_1) \leq 2/3$. And, similarly that $\mu(P_1) + \mu(P_1) =\rho(s, sf) + \rho(c, cf)  \geq \rho(f,csf) = \mu(P_3)$ implying that $\mu(P_1) \geq 1/3$.  We thus see that $\mu$ is an I-RUM only if $\mu(P_1) \in [1/3,2/3]$ (or equivalently $\mu(P_3) \in [1/3,2/3]$). It is also straightforward to check that all RUMs with $\mu(P_1)  \in [1/3,2/3]$ are I-RUMs.

Note that the preferences $P_1$ and $P_3$ are less "correlated" than $P_1$ and $P_2$ in the sense that $P_1$ and $P_2$ make the same choices in two menus, whereas $P_1$ and $P_3$ in none. This observation supports the idea that a less correlated set of preferences implies that a higher proportion of RUMs with support on this set have an irrational representation.
\end{customthm}

The correlation bounds capture the intuition of example \ref{ex1.1cc} when restricted to the set of RUMs. By rewriting the bounds in the space of the preferences, one can see that uncorrelated preferences and uncorrelated mistakes are observationally equivalent in the aggregate. For simplicity, we write $n(P,P')$ instead of $n(c_P, c_{P'})$ to refer to the measure of correlation between preferences.

\begin{lemma}[$\mathbb{C}^{\rho}_{P}$ within RUMs]  \label{lemma: Frum} Let $\rho$ be a RUM with distribution $\mu$ (i.e. $\rho=\rho_{\mu}$). Then for all preference $P$ it holds that $$\mathbb{C}^{\rho}_P = \sum_{P' \in \mathcal{P}} \mu(P')n(P,P').$$  \end{lemma}

\begin{proof} The proof follows as a corollary of Lemma \ref{lemma: CRNUM}. \end{proof}

The following result follows as a corollary to Theorem \ref{thm: RNUM}. 

\begin{theorem} \label{thm: main}
Let $\rho$ be a RUM, then it is an I-RUM if and only if it satisfies the correlation bounds.
\end{theorem}

The proof of Theorem \ref{thm: main} is a specialized version of that of Theorem \ref{thm: RNUM}. Nonetheless, it acquires some interesting features that are lost in the more general case. First, the proof of Theorem \ref{thm: main}, differently from Theorem \ref{thm: RNUM}, relies only on a small class of irrational choice functions (1-step away choice functions, i.e. those with $n(P, c)=K-2$ for some preference $P$), in fact so small that its share among all choice functions vanishes as the number $N$ of alternatives in $X$ grows.\footnote{See \cite{kalai}, \cite{giarlotta2023context} for a proof of this fact. More precisely, they show that the proportion of choice functions that can be rationalized by fewer than $N-1$ rationales tends to zero as $N$ tends to infinity. Since any 1-step away choice function can be rationalized by two rationales, the claim in the text follows.} In section \ref{sec: identification}, we will show how this feature creates an interesting avenue to discuss identification and comparative statics.  Second, the extreme points of the (convex) set of I-RUMs are now dual RUMs. 
Dual RUMs provide an interesting geometric visualization of the set of I-RUMs that we illustrate with a minor modification of our opening example.

\begin{customthm}{1}[continued]

Theorem \ref{thm: main} provides the visualization of the set of I-RUMs related to RUMs with support $P_1$, $P_2$, and $P_3$ as convex combinations of the respective irrational dual RUMs. In this simple example, the numbers $n(P_1,P_2) = n(P_2, P_3) = 1$, and $n(P_1, P_3) = n(P_3, P_2) = 0$ contain all the information needed. To illustrate, the dual RUM with support $P_1, P_2$ has the following correlation bounds (here $K-2=2$):
\[
\mathbb{C}_{P_{1}}^{\rho} - 2 = 3\mu(P_{1}) + \mu(P_{2}) - 2 \leq 0, \quad \mathbb{C}_{P_{2}}^{\rho} - 2 = \mu(P_{1}) + 3\mu(P_{2}) - 2 \leq 0,
\]
which imply $\mu(P_1) = \mu(P_2) = \frac12$. Instead, the dual RUM with support $P_1, P_3$ has the following correlation bounds:
\[
\mathbb{C}_{P_{1}}^{\rho} - 2 = 3\mu(P_{1}) + 0\mu(P_{3}) - 2 \leq 0, \quad \mathbb{C}_{P_{3}}^{\rho} - 2 = 0\mu(P_{1}) + 3\mu(P_{3}) - 2 \leq 0
\]
which imply $\mu(P_1) \in [\frac13, \frac23]$.

\begin{minipage}{0.45\linewidth}
\centering
\begin{table}[H]
\centering 
\begin{tabular}{c | c c c } 
\hline\hline 
 Menu & $\textcolor{black}{\max(A,P_1)}$ & $\max(A,P_2)$ & $\max(A,P_3)$ \\ [0.1ex] 
 \hline 
$csf$ & $c$ & $s$ & $f$ \\  
$cs$ & $c$  & $s$ & $s$ \\ 
$cf$ & $c$ & $c$ & $f$  \\ 
$sf$ & $s$ & $s$ & $f$ \\ 
\hline\hline 
\end{tabular}
\end{table} 
\end{minipage}
\hspace{0.1\linewidth}
\begin{minipage}{0.45\linewidth}
\centering
\begin{figure}[H] \label{fig: dualRUM}

\tikzset{every picture/.style={line width=0.75pt}} 

\begin{tikzpicture}[x=0.75pt,y=0.75pt,yscale=-1,xscale=1]

\draw   (314.38,77) -- (407.75,229.7) -- (221,229.7) -- cycle ;
\draw  [color={rgb, 255:red, 208; green, 2; blue, 27 }  ,draw opacity=1 ][fill={rgb, 255:red, 208; green, 2; blue, 27 }  ,fill opacity=0.2 ][line width=1.5]  (285.75,124.7) -- (362.75,156.7) -- (318.75,229.7) -- (275.75,229.7) -- (246.75,186.7) -- cycle ;

\draw (270,168) node [anchor=north west][inner sep=0.75pt]  [font=\scriptsize] [align=left] {$\displaystyle I-RUMs$};
\draw (407,240) node [anchor=north west][inner sep=0.75pt]   [align=left] {$\displaystyle P_{2}$};
\draw (233.23,166.87) node [anchor=north west][inner sep=0.75pt]  [font=\scriptsize,rotate=-302.63] [align=left] {$\displaystyle dual\ RUMs$};
\draw (364.54,120.58) node [anchor=north west][inner sep=0.75pt]  [font=\scriptsize,rotate=-57.77] [align=left] {$\displaystyle dual\ RUMs$};
\draw (285,250) node [anchor=north west][inner sep=0.75pt]  [font=\scriptsize] [align=left] {$\displaystyle dual\ RUMs$};
\draw (308,50) node [anchor=north west][inner sep=0.75pt]   [align=left] {$\displaystyle P_{1}$};
\draw (208,241) node [anchor=north west][inner sep=0.75pt]   [align=left] {$\displaystyle P_{3}$};

\end{tikzpicture}
\end{figure}
\end{minipage}
\end{customthm}

We conclude the section by providing two corollaries that even more intuitively convey the idea that I-RUMs are RUMs with sufficiently uncorrelated preferences. The first (which follows directly from Theorem \ref{thm: main}) shows that if a preference $P$ in the support of a RUM has a "sufficiently high" probability mass then the RUM cannot be an I-RUM.

\begin{corollary} Let $\rho$ be a RUM with distribution $\mu$.  If $\mu(P) > \frac{K-2}{K-1}$ for some preference $P$ then $\rho$ is not an I-RUM. \end{corollary}

The second, instead, shows that any RUM with sufficiently spread-out probability on the preferences is an I-RUM. 

\begin{corollary} Let $\rho$ be a RUM with distribution $\mu$. If $\mu(P) \leq \frac14$ for all $P \in \mathcal{P}$ then $\rho$ is an I-RUM. \end{corollary}

\begin{proof} We prove that $\rho$ satisfies the correlation bounds. Given $P$ let the preference $Q$ be the \textit{twin-preference} of $P$, namely the preference that agrees with $P$ on the ranking of all alternatives except their worst two alternatives. We rewrite the correlation bounds as in Lemma \ref{lemma: Frum} and note that we have:
 $$\mathbb{C}^{\rho}_P =\
\sum_{P' \in \mathcal{P}} \mu(P')n(P,P') = $$ $$\mu(P)(K-1) + \mu(Q)(K-1) + \sum_{P' \in \mathcal{P} \setminus \{P,Q\}} \mu(P')n(P,P') \leq  $$
 $$[\mu(P)+\mu(Q)](K-1) + \sum_{P' \in \mathcal{P} \setminus \{P,Q\}} \mu(P')(K-3) =$$ $$ [\mu(P)+\mu(Q)](K-1) + [1-\mu(P)-\mu(Q)](K-3) \leq K-2$$

where the first (in)equality follows by Lemma \ref{lemma: Frum}, the second (in)equality since $n(P, P)=n(P, Q)=K-1$, the third (in)equality follows since $n(P,{P'}) \leq K-3$ for all $P' \in \mathcal{P} \setminus \{P, Q\}$, the fourth in(equality) since $\sum_{P' \in \mathcal{P} \setminus \{P, Q\}} \mu(P')= 1- \mu(P) - \mu(Q)$ and the final inequality follows since $\mu(P)+\mu(Q) \leq \frac12$.\end{proof}

\subsection{Identification} \label{sec: identification}

Identification issues within RUMs are well-known and investigated \citep{barbera1986falmagne, fishburn1998, turansick, suleymanov2024branching}, and a natural question arises about whether irrational representations also suffer from identification issues. The answer seems obvious since there are many more irrational than rational choice functions. This intuition is generally correct;\footnote{Since the set of I-RUM representations is spanned by the set of preferences, and its extreme points are strictly more than the number of preferences (for non-trivial cases), I-RUM representations are generally rationalized by a multiple set of extreme points (see \cite{chambers2025limits} for a related discussion).} however, some stochastic choice functions that play an essential role in the proof of Theorem \ref{thm: RNUM} have a unique irrational representation.

Before stating the identification result, we need to introduce some further notation. Two choice functions $c,c'$ are equivalent w.r.t. preference $P$, denoted $c \sim_P c'$, if $c(A)=c'(A)$ for all $A \neq \{x_P(N-1),x_P(N)\}$. It is straightforward to check that $\sim_P$ is an equivalence relation and that each equivalence class of $\sim_P$ contains exactly two choice functions that agree everywhere except in $\{x_P(N),x_P(N-1)\}$. Denote the equivalence class containing the choice function $c$ by $[c]_P$. We say that a stochastic choice function $\rho$ has an \emph{essentially unique RNUM representation} if there is a preference $P$ such that for any two RNUM representations $\mu$ and $\mu'$ of $\rho$ it holds that $\mu([c]_P)=\mu'([c]_P)$ for all $c \in \mathcal{C}$. 

\begin{proposition} \label{pro: uniqueness}
If $\rho$ satisfies the correlation bounds with equality, then it has an essentially unique RNUM representation with support on the 1-step away choice functions.   \end{proposition}
\begin{proof}
See Appendix \ref{app: other}.
\end{proof}

Applying Proposition \ref{pro: uniqueness} to the set of RUMs leads to an interesting observation. If $\rho$ is an I-RUM and satisfies the correlation bounds with equality, the uniqueness of the irrational (RNUM) representation does not necessarily imply uniqueness of the rational (RUM) representation.\footnote{To show this, let $P_1$ and $P_2$ be preferences s.t. $a P_1 b P_1 c P_1 d$ and $b P_2 a P_2 d P_2 c$. Define a (dual) RUM $\mu$ by  $\mu(P_1) = 1/4$ and $\mu(P_2) = 3/4$. It is straightforward to verify that $\mu$ satisfies the correlation bounds with equality. Hence, by Proposition \ref{pro: uniqueness}, the RNUM representation is unique. Let $P'_1,P'_2$ be s.t. $a P'_1 b P'_1 d P'_1 c$ and $b P'_2 a P'_2 c P'_2 d$. It is then clear that $\rho$ can also be represented by $\mu'$ defined as $\mu'(P'_1)=\mu'(P'_2)=\frac14$ and $\mu'(P_2) = \frac12$, so the RUM representation is not unique.}  This observation holds even for the extreme points of the set of I-RUMs. As shown earlier, the extreme points of the set of I-RUMs are dual RUMs. \cite{mariottidual} shows that dual RUMs under weak conditions have a unique dual RUM representation. However, they do not necessarily have a unique RUM representation. This implies that even if the extreme points of the set of I-RUMs are dual RUMs, i.e. the heterogeneity in preferences is very limited, they are identified only "within irrational interpretations".

\section{Comparative statics on the fraction of irrational decision-makers: $\sigma$-Random Non-Utility Models}

So far, we have studied irrational representations either per s\'e (RNUMs) or paired with rational representations (I-RUMs). We now introduce and characterize a class of representations, $\sigma$-RNUMs, in which at least $\sigma$ decision-makers are irrational (RNUMs are an extreme case with $\sigma=1$). The subsequent results will allow us to state a comparative statics result on the maximal fraction $\sigma$ of irrational decision makers in an RCM representation of a stochastic choice function (see subsection \ref{sec: comp} for details).

\begin{definition}
A $\sigma$-RNUM is an RCM $\mu$ such that $\mu(\mathcal{C} \setminus \mathcal{P}) \geq \sigma$.
\end{definition}

We characterize $\sigma$-RNUMs by relaxing the correlation bounds.

\begin{definition}[$\sigma$-Correlation Bounds]
Let $\sigma$ $\in [0,1]$. A stochastic choice function $\rho$ satisfies the $\sigma$-\textit{Correlation Bounds} if for all preferences $P$ it holds that
$$\mathbb{C}^{\rho}_{P} - [K-2] \leq 1 - \sigma.\quad\text{}$$
\end{definition}

\begin{proposition} \label{pro: aRUM}
A stochastic choice function $\rho$ is a $\sigma$-RNUM if and only if it satisfies the $\sigma$-correlation bounds.
\end{proposition}

We split the proof into two lemmas. The first lemma implies necessity, whereas the second lemma implies sufficiency of the $\sigma$-correlation bounds. This expositional choice is because the two directions convey different points. On the one hand, the contrapositive of the necessity statement provides a lower bound on the fraction of rational decision-makers, a result that requires independent discussion. On the other hand, the sufficiency part of proposition \ref{pro: aRUM} suggests a comparative statics result between the correlation bounds and the fraction of irrational decision-makers.

\begin{lemma}[Necessity] \label{lemma: fraction}
If $\mathbb{C}_{P}^{\rho} - [K-2] \geq 1- \sigma$ for some preference $P$ then for all RCMs $\mu$ such that $\rho_{\mu} = \rho$ the fraction of rational decision-makers is at least $1-\sigma$.
\end{lemma}
\begin{proof}
The proof follows immediately from the Fr\'echet bounds in Lemma \ref{lemma: strict}. To see this, let $P$ be a preference such that $\mathbb{C}^{\rho}_{P} - [K-2] \geq 1 - \sigma$. Let $Q$ be the "twin-preference" of $P$. Then lemma \ref{lemma: strict} implies that $$\mu( \{c_P,c_Q\} ) \geq \max\{0,  \sum_{A \in \mathcal{A}(P)} \rho(c_P(A),A) - (K-2)\} \ge 1- \sigma$$ for all RCMs $\mu$ with $\rho=\rho_{\mu}$.
\end{proof}

To see that the $\sigma$-correlation bounds are necessary for a $\sigma$-RNUM representation, assume that $\mathbb{C}^{\rho}_P - [K-2] > 1-\sigma$ for some preference $P$. Then  $\mathbb{C}^{\rho}_P - [K-2] = 1-\sigma + s$ for some $s > 0$. Lemma \ref{lemma: fraction} then implies that every RCM representation of $\rho$ has at least $1-\sigma + s$ rational decision-makers, and so cannot be a $\sigma$-RNUM. Lemma \ref{lemma: fraction} provides a \emph{lower bound} on the fraction of rational decision-makers in any RCM representation of a given $\rho$. Being a lower bound, this is a rather conservative measure. In reality, the fraction of rational decision-makers is likely higher. To see this, one could obtain an upper bound on the fraction of rational decision-makers by performing a similar separation exercise as in \cite{apesteguia2021separating}. I.e. by finding the maximal fraction $\alpha$ such that $\rho= \alpha \rho_{\mu} + (1-\alpha) \rho^{'}$ for some RUM $\mu$.\footnote{The problem of finding such a maximal fraction is not an easy one. \cite{apesteguia2021separating} provide a method for SCRUMs \citep{apesteguia2017single}, while for generic RUMs, this remains an open problem we do not tackle.} The "true" fraction of rational individuals is probably between these extremes.

\begin{lemma}[Sufficiency] \label{lemma: suff}
If $\mathbb{C}_{P}^{\rho} - [K-2] \leq 1- \sigma$ for all preferences $P$ then $\rho$ is a $\sigma$-RNUM.
\end{lemma}


\begin{proof}
Let $P$ be a preference that maximizes $\mathbb{C}^{\rho}_{P}$, let $Q$ be its twin, and set $m=\mathbb{C}^{\rho}_{P}-[K-2]\leq 1-\sigma$. If $m\leq 0$ then $\rho$ is an RNUM by Theorem \ref{thm: RNUM}. Thus, assume that $m>0$ and $\sigma>0$ (if $\sigma=0$ there is nothing to prove), so that $m\in(0,1)$. Let $\mu$ be an RCM representation of $\rho$. Then by Lemma \ref{lemma: CRNUM}, since $n(P,c)\leq K-2$ for all $c\notin\{c_P,c_Q\}$, we have $[K-2]+m=\mathbb{C}^{\rho}_{P}=\sum_{c\in\mathcal{C}}\mu(c)n(P,c)\leq [K-2]+\mu(\{c_P,c_Q\})$, so that $\mu(\{c_P,c_Q\})\geq m>0$. Define a dual RUM $\rho^d$ for all $a \in A$ and $A \in \mathcal{A}$ by $$\rho^d(a,A) = \frac{\mu(c_P)}{\mu(\{c_P,c_Q\} )} \mathbf{1}\{a= \max(P,A)\} +\frac{\mu(c_Q)}{\mu(\{c_P,c_Q\} )}  \mathbf{1}\{a= \max(Q,A)\},$$ and define $\hat{\rho}(a,A) =  \frac{\rho(a,A) - m \rho^{d}(a,A)}{1-m}$ for all $a \in A$ and $A \in \mathcal{A}$, which is non-negative since $\mu(\{c_P,c_Q\})\geq m$. Then $\rho= m \rho^{d} + (1-m) \hat{\rho}$. Further, we note that
$$ \mathbb{C}^{\hat{\rho}}_{P} = \frac{1}{1-m} \sum_{A \in \mathcal{A}(P)} [\rho(c_P(A),A) - m \rho^{d}(c_P(A),A)] =$$
$$ = \frac{1}{1-m} [\mathbb{C}^{\rho}_P - m[K-1]] =
\frac{1}{1-m} [m + [K-2] - m[K-1]]  =   K-2.$$
For every preference $P'\notin\{P,Q\}$ and every $c\in\mathcal{C}$ we have $n(P',c)+n(P,c)\leq 2[K-2]$, since $n(R,c)=K-1$ only if $c$ is $c_R$ or the twin of $c_R$, and $n(R,c_{R'})\leq K-3$ whenever $R,R'$ are not twins (as in the proof of Corollary 4.5). Hence, by Lemma \ref{lemma: CRNUM}, $\mathbb{C}^{\hat{\rho}}_{P'}\leq 2[K-2]-\mathbb{C}^{\hat{\rho}}_{P}=K-2$. Hence, it follows that $\hat{\rho}$ satisfies the correlation bounds, and by Theorem \ref{thm: RNUM} it has an RNUM representation. Let $\mu^d$ and $\hat{\mu}$ be such that $\rho^d=\rho_{\mu^d}$ and $\hat{\rho}= \rho_{\hat{\mu}}$, with $\hat{\mu}$ an RNUM. Then it is clear that $\rho= \rho_{m\mu^d + (1-m) \hat{\mu}}$ and hence that $\rho$ is a $(1-m)$-RNUM, hence a $\sigma$-RNUM. \end{proof}

Note that $\mathbb{C}^{\rho}_P$ can be viewed as a measure of the correlation between $\rho$ and preference $P$. Thus, Lemma \ref{lemma: suff} says that if the correlation between $\rho$ and every rational decision maker (normalized by $K-2$) is lower than $1-\sigma$, then at least a fraction of $\sigma$ decision-makers must be irrational. The preceding result also suggests that the lower the correlation between $\rho$ and any preference $P$ is (so $\rho$ satisfies the $\sigma$-correlation bounds for a higher $\sigma$), the larger the fraction $\sigma$ of irrational decision-makers will be. In the next subsection, we discuss comparative statics results along these lines. 

\subsection{Comparative statics on the fraction of irrational decision-makers} \label{sec: comp}

We are now ready to state our comparative statics result on the fraction of irrational decision-makers. For each stochastic choice function $\rho$ define $$\mathbb{I}(\rho) =1- \max\Big\{0,\ \max_{P \in \mathcal{P}} \big(\mathbb{C}^{\rho}_P - [K-2]\big)\Big\}.$$ Note that $\mathbb{I}(\rho)$ measures the correlation between $\rho$ and the irrational decision-makers: the higher $\mathbb{I}(\rho)$ is, the lower is $\max_{P \in \mathcal{P}} \mathbb{C}^{\rho}_P$, i.e. the lower is the correlation between $\rho$ and the rational decision-makers. We next show that $\mathbb{I}(\rho)$ can also be viewed as a measure of irrationality, in the sense that it equals the maximal fraction of irrational decision-makers compatible with $\rho$.

\begin{proposition} \label{pro: irrcomp} Let $\rho$ be a stochastic choice function then $$\max_{\substack{\mu \in \Delta(\mathcal{C}) \\ \rho=\rho_{\mu}}} \mu(\mathcal{C}\setminus\mathcal{P}) =\mathbb{I}(\rho),$$ i.e. $\mathbb{I}(\rho) $ is the maximal fraction of irrational decision-makers in any RCM representation $\mu$ of $\rho$. \end{proposition}

The preceding result thus implies that if $\rho$ and $\rho'$ are stochastic choice functions such that $\mathbb{I}(\rho) > \mathbb{I}(\rho')$ then $\rho$ is "more irrational" than $\rho'$ in the sense that it can be represented by a larger fraction of irrational decision-makers. The proof of proposition \ref{pro: irrcomp} is straightforward and follows directly from the characterization of $\sigma$-RNUMs. To see this, set $\sigma$ $= \mathbb{I}(\rho)$ and note that $\mathbb{C}^{\rho}_P - [K-2] \leq 1- \mathbb{I}(\rho)= 1-\sigma$ for all preferences $P$. Hence $\rho$ has a $\sigma$-RNUM representation with a fraction of $\mu(\mathcal{C}\setminus\mathcal{P}) \geq \sigma = \mathbb{I}(\rho)$ irrational decision makers. Next, note that if $\max_{P} \mathbb{C}^{\rho}_P - [K-2] \leq 0$ then $\mathbb{I}(\rho)=1$ and there is nothing to prove; otherwise $1- \mathbb{I}(\rho) = \max_{P \in \mathcal{P}} \big(\mathbb{C}^{\rho}_P -[K-2]\big)$, and lemma \ref{lemma: fraction} implies that every RCM representation $\mu$ of $\rho$ has at least a fraction $1-\mathbb{I}(\rho)$ of rational decision-makers, hence at most a fraction $\mathbb{I}(\rho)$ of irrational decision-makers.

\section{Discussion I - Non-falsifiability of rationality under RUM} \label{sec: discussion}

In Theorem \ref{thm: main}, we look at irrational representations within the set of RUMs. An implication is that many RUMs are described by populations of irrational individuals. I.e. the RUM hypothesis is valid in the aggregate but fails at the individual level. This is an extreme case of what we call the non-falsifiability of the rational interpretation of the RUM hypothesis. However, there are many other cases where this may be an issue, such as when only a fraction of the population is irrational. Our results on $\sigma$-RNUMs allow us to study these "intermediate" cases of non-falsifiability. We note that for this section, and related proofs, we assume observability of all menus with cardinality at least two, as some results rely on the characterization of the RUM via the Block-Marschak polynomials, which requires full domain.

\subsection{The limits of non-falsifiability.} In the section we show exactly when non-falsifiability \emph{may} be an issue for RUM, in the sense that a given RUM is also explained by a fraction of irrational individuals (i.e. an RCM $\mu$ with $\mu(\mathcal{C}\setminus \mathcal{P}) >0$).

\begin{proposition} \label{pro: piRUM} \leavevmode Let $\rho$ be RUM. Then the following claims are equivalent:
\begin{enumerate}
    \item \label{fals1} $\rho$ has an RCM representation $\mu$ with $\mu(\mathcal{C} \setminus \mathcal{P}) > 0$.
        \item \label{fals2} $\rho$ has a $\sigma$-RNUM representation for some $\sigma$ $> 0$.
          \item  \label{fals3} $\mathbb{C}^{\rho}_P - [K-2] <1$ for all preferences $P$.
          \item \label{fals4} $\rho$ is a non-trivial convex combination of an I-RUM and a dual RUM.
\end{enumerate}

\end{proposition}
\begin{proof}
We first show that (\ref{fals1}) implies  (\ref{fals2}). Let $\rho$ have an RCM representation  $\mu$ with $\mu(\mathcal{C} \setminus \mathcal{P}) > 0$. Set $\sigma$ $=\mu(\mathcal{C} \setminus \mathcal{P}) >0$ then clearly $\rho$ has a $\sigma$-RNUM representation. Next, we show that (\ref{fals2}) implies  (\ref{fals3}). If $\rho$ has a $\sigma$-RNUM representation then it satisfies the $\sigma$-correlation bounds for some $\sigma$ $>0$, hence $\mathbb{C}^{\rho}_P - [K-2] \leq 1-\sigma <1$ for all preferences $P$. The implication (\ref{fals3}) $\Rightarrow$ (\ref{fals4}) follows by inspection of Lemma \ref{lemma: suff}. 
 Finally, it is immediate that (\ref{fals4}) implies (\ref{fals1}).\end{proof}

The first and most immediate implication of Proposition \ref{pro: piRUM} is that almost all RUMs are non-falsifiable in that they are compatible with a fraction of irrational decision-makers. In fact, this is true not only for every RUM with full support, but also for every RUM that satisfies positivity, such as the Logit model. A second, more subtle implication is that every RUM can be decomposed into an irrational component and a minimally rational core, since it can be written as a convex combination of an irrational representation (I-RUM) and a rational one (dual RUM). 


Now, we show that (i) irrational behavior may be unconstrained, and (ii) the fraction of irrational decision-makers compatible with the RUM hypothesis is monotonically related to the set of potential irrational behaviors.

\begin{proposition} \label{pro: falsifiability} Let $\rho^*$ be a RUM with full support and $\mathtt{M}$ a collection of RCMs. Then there is a $\bar{\alpha}(\mathtt{M}) \in (0,1)$ such that:
\begin{enumerate}
    \item if $1 > \alpha \geq \bar{\alpha}(\mathtt{M})$ then for all $\mu \in \mathtt{M}$ it holds that $\alpha \rho^* + (1-\alpha) \rho_{\mu}$ is a RUM.
    \item if $0 < \alpha < \bar{\alpha}(\mathtt{M})$ then there is a $\mu \in \mathtt{M}$ such that $\alpha \rho^* + (1-\alpha) \rho_{\mu}$ is not a RUM.
\end{enumerate}
\end{proposition}
\begin{proof}
The proof is in Appendix \ref{app: other}.\footnote{The overall idea of this result is to counterweight a $\rho$ that is not a RUM, i.e. it has negative BM polynomials, with a $\rho^*$ that is a RUM with full support, i.e. it has strictly positive BM polynomials (everywhere). In Appendix \ref{app: BM}, we discuss the standard characterization of RUMs with BM polynomials, as well as some preliminary results that play a role in the proofs of the results in this section. 
}
\end{proof}

Proposition \ref{pro: falsifiability} shows that whenever the rational decision-makers are represented by a RUM with full support, then there exists a lower bound on the fraction $\bar{\alpha}(\mathtt{M})$ of rational decision-makers such that the RUM hypothesis is satisfied irrespective of the behavior of the remaining decision-makers. We present an illustration of Proposition \ref{pro: falsifiability} within our opening example.

\begin{customthm}{1}[continued] \label{ex1.1c}

Consider a population where a fraction of decision-makers are rational and, differently from above, have a uniform distribution on the set of preferences on $\{ c, s, f \}$; while the remaining ones are irrational. Simple calculations show that whenever at least $\bar{\alpha}(\mathtt{M}) = 6/7$ of the decision-makers are rational, then the behavior of the remaining irrational ones is unconstrained w.r.t. stochastic rationality.\footnote{To see this, note that regularity is a necessary and sufficient condition for RUM with only three alternatives \citep{block}. Suppose $\mu$ is the proportion of decision-makers who behave rationally with a uniform distribution on the set of preferences, and $1-\mu$ are those who behave irrationally. Assuming all irrational decision-makers make the same mistake, regularity is satisfied whenever $1/2 \mu - 1/3 \mu + \mu \geq 1$, or equivalently, $\mu \geq 6/7$.} Proposition \ref{pro: falsifiability} also shows that the lower bound is tight, i.e. whenever there are less than $\bar{\alpha}(\mathtt{M})$ rational decision-makers then there is a group of $1-\bar{\alpha}(\mathtt{M})$ irrational decision-makers that will induce a rejection of the RUM hypothesis. Again, if $\alpha < 6/7$ then if $\mu$ is the RCM where all decision-makers choose $c$ from $csf$ and $s$ from $cs$, then $\alpha \rho^* + (1- \alpha) \rho_{\mu}$ is not a RUM.  
\end{customthm}

The bound defined in Proposition \ref{pro: falsifiability} may seem extreme at first glance, however, it assumes the extreme scenario in which every irrational decision-maker makes the same mistake. To see how the bound changes with the set of potential irrational behaviors, if $c_1$ is at most twice as likely to appear as $c_2$, the fraction of irrational decision-makers compatible with the RUM hypothesis increases from 1/7 to 1/3, while, as shown in our opening example, if $c_1$ and $c_2$ are equally probable this fraction becomes one.\footnote{Let $\mu$ be the proportion of decision-makers who switch to steak when frogs' legs are available. This implies that $1/2 \mu$ decision-makers have the opposite irrational behavior and $1 - 3/2 \mu$ decision-makers are rational with uniform preferences. The analogous regularity condition (now for $s$) yields $\mu \leq 2/9$, or equivalently, a proportion of rational decision-makers greater or equal to $2/3$.}

The following corollary formalizes the relationship between our bound and the behavior of irrational decision-makers.

\begin{corollary} \label{cor: Ircor}
Let $\mathtt{M}, \mathtt{M}'$ be two collections of RCMs, and $\rho^*$ a RUM with full support. Then $\mathtt{M}' \subseteq \mathtt{M}$ implies $\bar{\alpha}(\mathtt{M}') \leq \bar{\alpha}(\mathtt{M})$.
\end{corollary}

A brief comment w.r.t. the recent literature is now required as both Proposition \ref{pro: falsifiability} and Corollary \ref{cor: Ircor} remind the approach of \cite{apesteguia2021separating}. The authors aim to separate the randomness of a stochastic choice function that is aligned with a model, i.e. a subset of stochastic choice functions, from residual behavior, i.e. randomness that cannot be explained by the model, with the purpose of "explaining the largest possible fraction of data using the model". To be a non-trivial exercise, the stochastic choice function that the authors focus on has to be outside the model. By studying representations of RUMs with irrational decision-makers, we show that even if a stochastic choice function $\rho$ is a RUM, i.e. within the model, we can construct a separation between a rational representation and an irrational one with the latter being fully unstructured as in \cite{apesteguia2021separating}.

\subsection{Every RUM is asymptotically arbitrarily close to an I-RUM}

We have shown that the non-falsifiability issue, i.e. the existence of $\sigma$-RNUM representations, is widespread and that often irrational behaviors are unconstrained. We conclude this section by showing that irrational representations are paramount "in the limit". Specifically, we show that all RUMs can be approximated\footnote{For convenience, we use the maximum norm restricted to the set of stochastic choice functions (but the result holds w.r.t. any norm on $\mathbb{R}^d$ since all norms are equivalent on $\mathbb{R}^d$).} arbitrarily close by an I-RUM when the set of alternatives $X$ is sufficiently large.

 \begin{proposition} \label{pro: almost} For every RUM $\rho$  there is an I-RUM $\rho'$ with $$\left\lVert \rho-\rho' \right\rVert \leq \frac{1}{2^{|X|}-|X|-2}.$$ \end{proposition}
 \begin{proof}
The proof is in Appendix \ref{app: other}.
 \end{proof}
 
 For instance, if there are exactly five alternatives in $X$, i.e. $|X| = 5$, then $2^{|X|}-|X|-2=2^5-5-2= 25$, the distance from a RUM to the set of I-RUMs is less than $1/25$. This bound decreases exponentially as the number of alternatives increases.

\section{Discussion II - Individual vs stochastic rationality} \label{subsec: paradox}

By being valid for all aggregate data, Proposition \ref{pro: aRUM} reveals the tension between the notions of stochastic and individual rationality. On the one hand, there are stochastically rationalizable choices with an irrational representation. On the other hand, there are stochastic choice functions that are not stochastically rational but that lack an irrational representation. The paradoxical conclusion that emerges is thus that some non-rationalizable stochastic choice functions may need more rational choice functions than some rationalizable ones do.

We clarify this point within example \ref{ex1.1}. Consider the case where $\rho(c,cf) = \rho(s,sf) = 1$, $\rho(f,csf) = 0$, i.e. aggregate choices are highly concentrated. The set of stochastic choice functions that satisfy these constraints can be represented in a 2-dimensional space as in Figure \ref{fig: example} below. Here, the set of RUMs is characterized by the equation $\rho(c, csf)=\rho(c, cs)$ and the set of RNUMs by the equation $\rho(c, csf) = 1 - \rho(c, cs)$. This implies, as we know, that the only I-RUM is the point $\rho(c, csf) = \rho(c, cs) = \frac12$. Note that, a minor perturbation $\varepsilon>0$ such that $\rho^{\varepsilon}(c, csf) = \frac12 + \varepsilon$ implies both that $\rho^{\varepsilon}$ is not a RUM because regularity is violated and also that $\rho^{\varepsilon}$ is not an RNUM, and so cannot be represented by a population of irrational decision-makers because the correlation bounds are violated (see the blue point in the middle of Figure \ref{fig: example}).

\begin{figure}[H]
    \centering

\tikzset{every picture/.style={line width=0.75pt}} 
\begin{tikzpicture}[x=1pt,y=1pt,yscale=-1,xscale=1]

\draw   (255.3,87) -- (400,87) -- (400,231.7) -- (255.3,231.7) -- cycle ;
\draw [color={rgb, 255:red, 74; green, 144; blue, 226 }  ,draw opacity=1 ]   (394.75,101.7) ;
\draw [shift={(394.75,101.7)}, rotate = 0] [color={rgb, 255:red, 74; green, 144; blue, 226 }  ,draw opacity=1 ][fill={rgb, 255:red, 74; green, 144; blue, 226 }  ,fill opacity=1 ][line width=0.75]      (0, 0) circle [x radius= 3.35, y radius= 3.35]   ;
\draw [shift={(394.75,101.7)}, rotate = 0] [color={rgb, 255:red, 74; green, 144; blue, 226 }  ,draw opacity=1 ][fill={rgb, 255:red, 74; green, 144; blue, 226 }  ,fill opacity=1 ][line width=0.75]      (0, 0) circle [x radius= 3.35, y radius= 3.35]   ;
\draw [color={rgb, 255:red, 208; green, 2; blue, 27 }  ,draw opacity=1 ][fill={rgb, 255:red, 208; green, 2; blue, 27 }  ,fill opacity=1 ]   (327.65,159.35) ;
\draw [shift={(327.65,159.35)}, rotate = 0] [color={rgb, 255:red, 208; green, 2; blue, 27 }  ,draw opacity=1 ][fill={rgb, 255:red, 208; green, 2; blue, 27 }  ,fill opacity=1 ][line width=0.75]      (0, 0) circle [x radius= 3.35, y radius= 3.35]   ;
\draw [shift={(327.65,159.35)}, rotate = 0] [color={rgb, 255:red, 208; green, 2; blue, 27 }  ,draw opacity=1 ][fill={rgb, 255:red, 208; green, 2; blue, 27 }  ,fill opacity=1 ][line width=0.75]      (0, 0) circle [x radius= 3.35, y radius= 3.35]   ;
\draw [color={rgb, 255:red, 208; green, 2; blue, 27 }  ,draw opacity=1 ][line width=0.75]    (255.3,231.7) -- (400,87) ;
\draw [color={rgb, 255:red, 74; green, 144; blue, 226 }  ,draw opacity=1 ]   (338.75,157.7) ;
\draw [shift={(338.75,157.7)}, rotate = 0] [color={rgb, 255:red, 74; green, 144; blue, 226 }  ,draw opacity=1 ][fill={rgb, 255:red, 74; green, 144; blue, 226 }  ,fill opacity=1 ][line width=0.75]      (0, 0) circle [x radius= 3.35, y radius= 3.35]   ;
\draw [shift={(338.75,157.7)}, rotate = 0] [color={rgb, 255:red, 74; green, 144; blue, 226 }  ,draw opacity=1 ][fill={rgb, 255:red, 74; green, 144; blue, 226 }  ,fill opacity=1 ][line width=0.75]      (0, 0) circle [x radius= 3.35, y radius= 3.35]   ;
\draw  [dash pattern={on 0.84pt off 2.51pt}] (308.45,159.35) .. controls (308.45,148.75) and (317.05,140.15) .. (327.65,140.15) .. controls (338.25,140.15) and (346.85,148.75) .. (346.85,159.35) .. controls (346.85,169.95) and (338.25,178.55) .. (327.65,178.55) .. controls (317.05,178.55) and (308.45,169.95) .. (308.45,159.35) -- cycle ;
\draw    (311,123) -- (324.84,149.92) ;
\draw [shift={(325.75,151.7)}, rotate = 242.8] [color={rgb, 255:red, 0; green, 0; blue, 0 }  ][line width=0.75]    (6.56,-1.97) .. controls (4.17,-0.84) and (1.99,-0.18) .. (0,0) .. controls (1.99,0.18) and (4.17,0.84) .. (6.56,1.97)   ;
\draw [color={rgb, 255:red, 74; green, 144; blue, 226 }  ,draw opacity=1 ][line width=0.75]    (255.3,87) -- (400,231.7) ;

\draw (237,176) node [anchor=north west][inner sep=0.75pt]  [font=\footnotesize,rotate=-270] [align=left] {$\displaystyle \rho ( c,cs)$};
\draw (310,235) node [anchor=north west][inner sep=0.75pt]  [font=\footnotesize] [align=left] {$\displaystyle \rho ( c,csf)$};
\draw (243,233) node [anchor=north west][inner sep=0.75pt]  [font=\scriptsize] [align=left] {$\displaystyle 0$};
\draw (243,80) node [anchor=north west][inner sep=0.75pt]  [font=\scriptsize] [align=left] {$\displaystyle 1$};
\draw (402,234.7) node [anchor=north west][inner sep=0.75pt]  [font=\scriptsize] [align=left] {$\displaystyle 1$};
\draw (338.12,125.69) node [anchor=north west][inner sep=0.75pt] [font=\footnotesize,rotate=-315] [align=left] {RUMs};

\draw (286.75,108.7) node [anchor=north west][inner sep=0.75pt] [font=\scriptsize] [align=left] {I-RUMs};
\draw (284.91,125.94) node [anchor=north west][inner sep=0.75pt] [font=\scriptsize,rotate=-45] [align=left] {\text{RNUMs}};

 \end{tikzpicture}
    \vspace{-1em}
    \caption{}
        \label{fig: example}
\end{figure}
\vspace{-0.5em}

The above reasoning can be pushed to a limit by letting $\rho(c, csf) > \rho(c, cs)$ and both  tend to one (see the blue point in the top-right of Figure \ref{fig: example}). This stochastic choice function violates the RUM hypothesis, while almost all decision-makers are individually rational. This example highlights the following point (which can be generalized to other examples). If the analyst observes a stochastic choice function with concentrated probability mass that is not a RUM then she is certain that: (i) both a RUM and an RNUM representation do not exist; (ii) a fraction of decision-makers are individually irrational; and (iii) a fraction of decision-makers are individually rational.

\subsection{An example within the theory of demand}

We apply the above reasoning to the Theory of Demand using an example by \cite{kitamura2018nonparametric}. We do not claim novelty for the intuitions that follow as among others \cite{im2022non} pointed out similar facts. However, we would like to re-interpret this example, in which correlation bounds are necessary and sufficient \citep{matzkin2007heterogeneous}, within our framework to provide some further insights. Figure \ref{fig:KS} displays the potential choices of decision-makers from two budget sets: $\mathcal{B}_1, \mathcal{B}_2$.

Assuming Walras' Law, $\pi_{i|j}$ denotes the share of the decision-makers who choose in the segment $i$ from the budget set $\mathcal{B}_j$. Ignoring the intersection, the choice probabilities in this example are described by the vector $(\pi_{1|1},\pi_{2|1},\pi_{1|2},\pi_{2|2})$. There are four possible pairs of choices that the decision-makers can make and only one of them is irrational, i.e. $\pi_{1|1}, \pi_{1|2}$. By similar reasoning as in our example \ref{ex1.1}, this setting does not admit I-RUMs. Nonetheless, our discussion easily follows. 

\begin{figure}[H]
    \centering
    \includegraphics[width=0.45\textwidth]{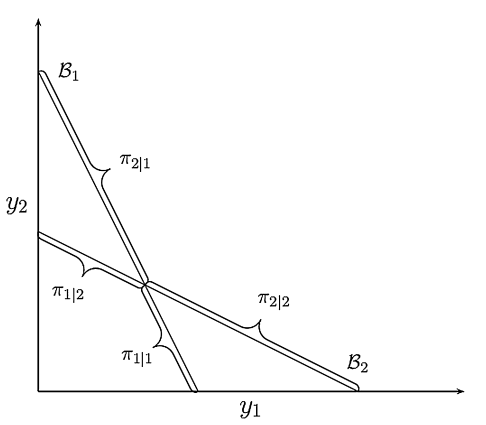}
    \vspace{-1em}
    \caption{Example from \cite{kitamura2018nonparametric}}
    \label{fig:KS}
\end{figure}

First, suppose we observe the stochastically rational choice probabilities $(\pi_{1|1},\pi_{2|1},\pi_{1|2},\pi_{2|2}) = (0.5, 0.5, 0.5, 0.5)$. This RUM has a $\sigma$-RNUM representation where the fraction of irrational decision-makers is at most 50\%. The contingency tables below display two possible scenarios in which either 0\% or 50\% of the decision-makers are irrational. One can also note that there is only one representation in which every decision-maker is rational, while there is an infinity of representations in which at least some decision-makers are irrational. In the contingency tables, we color the irrational choices in \textcolor{red}{red} and the rational ones in \textcolor{blue}{blue}. 

$$ \left[ \begin{array}{c| cc | c}
 & \pi_{1|1} & \pi_{2|1} &  \\ \hline
\pi_{1|2} & \textcolor{red}{0.5} & \textcolor{blue}{0} & 0.5 \\
\pi_{2|2} & \textcolor{blue}{0} & \textcolor{blue}{0.5} & 0.5 \\ \hline
 & 0.5 & 0.5 & 1
\end{array} \right]  \left[ \begin{array}{c| cc | c}
 & \pi_{1|1} & \pi_{2|1} &  \\ \hline
\pi_{1|2} & \textcolor{red}{0} & \textcolor{blue}{0.5} & 0.5 \\
\pi_{2|2} & \textcolor{blue}{0.5} & \textcolor{blue}{0} & 0.5 \\ \hline
 & 0.5 & 0.5 & 1
\end{array} \right] $$

Imagine now to observe the following non-stochastically rational choice probabilities: $(\pi_{1|1},\pi_{2|1},\pi_{1|2},\pi_{2|2}) = (1, 0, 0.1, 0.9)$. A rapid inspection will reveal that the fraction of irrational decision-makers is both at least and at most 10\% and that the unique RCM that produces the choice probabilities is:

$$ \left[ \begin{array}{c| cc | c}
 & \pi_{1|1} & \pi_{2|1} & \\ \hline
\pi_{1|2} & \textcolor{red}{0.1} & \textcolor{blue}{0} & 0.1 \\
\pi_{2|2} & \textcolor{blue}{0.9} & \textcolor{blue}{0} & 0.9 \\ \hline
 & 1 & 0 & 1
\end{array} \right] $$

This simple example reveals a couple of interesting facts. First, statistical tests that focus solely on the distance between the observed aggregate choice probabilities and the set of RUMs may lose relevant information regarding the rationality of the underlying behavior of the decision-makers. Second, as noticed by \cite{im2022non}, assuming individual rationality implies the identification of 50\% of the population with a high (resp. low) marginal rate of substitution between the two goods while relaxing this assumption opens to the possibility (identification of the joint distribution is now lost) that 100\% of the individuals has a mild marginal rate of substitution with 50\% of them being rational and 50\% irrational. Clearly, welfare conclusions may differ substantially between these two scenarios. This example conveys a general message, namely that normative conclusions, e.g. welfare assessments, are never only data-driven (or theory-free) but always rely on the analyst interpretation, in our case through different representations. Different interpretations may lead to starkly different normative conclusions which implies, for example, that the prior knowledge of the analyst about individual rationality is of crucial importance.

\section{Two final Extensions} \label{sec: extension}

\subsection{Ordered domains}

As aforementioned, Theorem \ref{thm: main} relies on the non-trivial detail that the universe of irrational behaviors is much larger than that of rational ones. The reader may, therefore, question whether our results would follow under some restrictions on irrational behaviors. As mentioned in Section \ref{sec: main}, a closer look at our proof of Theorem \ref{thm: main} already reveals a first answer, as our construction only relies on irrational choice functions that are 1-step away from the rational ones. This implies that irrational representations of RUMs only rely on a small class of irrational choice functions, whose share among all choice functions tends to 0 as $N$ goes to infinity.\footnote{Interesting discussions on the size of boundedly rational models can be found in \cite{kalai}, \cite{giarlotta2022bounded}, \cite{giarlotta2023context}.}

Recent influential literature has investigated other potential restrictions on the heterogeneity of rational and (potentially) irrational behaviors. \cite{apesteguia2017single} studied a subset of RUMs known as Single-Crossing RUMs while \cite{filiz2023progressive} studied a subset of RCMs known as Progressive Random Choice. We refer to these papers for formal definitions. Focusing on subsets of preferences or choice functions that can be ordered along one dimension, these papers provided unique identification results that, interestingly, seem to clash with our indeterminacy result (see \cite{chambers2025ordered} for a discussion on identification results and Fr\'echet bounds). We show that limiting the preference heterogeneity in RUMs to one dimension does not affect our results, while limiting the heterogeneity of the mistakes "almost" fully limits the explanatory power of irrational representations.

\begin{observation}
\leavevmode
\begin{enumerate}
    \item If $\rho$ is a full-support I-RUM then its irrational representation is not progressive.\footnote{The proof is simple. Suppose the support of an irrational representation $\mu$ is progressive w.r.t. to the strict total order $\triangleright$, sort the choice functions $\{ c_1, \dots, c_T \}$. By irrationality of $c_1$, there are menus $B  \subseteq A \subseteq X$ and distinct alternatives $a,b \in B$ with $c_1(A)=a$  and $c_1(B)=b$. This implies $\rho_{\mu}(a,A), \rho_{\mu}(b,B) > 0$. If $b \triangleright a$ then $\rho_{\mu}(a,B) = 0$ violating regularity and so $\rho_{\mu}$ is not a RUM. If $a \triangleright b$ then $\rho(b,A) = 0$ and the I-RUM is not full-support.}

    \item $\rho$ is an I-SCRUM if and only if it is a SCRUM and satisfies the correlation bounds.
    \item Almost all I-RUMs (I-SCRUMs) do not have a progressive irrational representation.
\end{enumerate}
\end{observation}

\subsection{The Logit Model}

Our final focus is on the most influential among the RUMs, the Logit Model \citep{luce59} which is commonly considered a canon of stochastic rationality \citep{cerreia2021canon}. We show that our main result can be straightforwardly translated to the Logit Model exploiting the existence of a utility function, hence an ordering of the alternatives. Call a stochastic choice function $\rho$ aligned w.r.t. a preference $P$ if for all menus $A$: $\rho(a,A) > \rho(b,A)$ implies $aPb$. The following observation suggests that in the case of aligned stochastic choice functions the correlation bounds only bind at the relevant preference $P$ as for all menus $A$ and $P' \not= P$, it holds that $\rho(c_P(A), A) \geq \rho(c_{P'}(A), A)$ since $c_P(A) P  c_{P'}(A)$.

\begin{observation}
Let $\rho$ be aligned w.r.t. a preference $P$. Then $\rho$ has an irrational representation if and only if $\mathbb{C}^{\rho}_P - [K-2] \leq 0$.
\end{observation}

This observation naturally extends to the Logit Model by letting $P$ be a strict extension of the utility function. A consequence is that a sufficient condition for a Luce stochastic choice function to have an irrational representation is: $u(a) \leq (K-2)u(b)$ for all $a,b \in X$ (the proof follows almost immediately from the above observation and it is therefore omitted). This condition intuitively shows that even the strictest notion of stochastic rationality is susceptible to irrational representations whenever the utilities are not sufficiently far apart. Finally, we conclude by noticing that the present argument extends to other (aligned) influential models such as the additive perturbed utility model \citep{fudenberg} and the simple scalability model \citep{tversky1972choice}.

\section{Concluding remarks}

In the last decades, economics has welcomed several critiques of its rational foundations. Being located within the literature on individual decision-making, our paper fits within this tendency as we have provided a characterization of the set of populations that are rational at an aggregate level and can be also represented as irrational ones at an individual level. 

In these concluding remarks, we touch on a different and more positive viewpoint of our paper borrowed from \cite{becker}: "Economic theory is much more compatible with irrational behavior than had been previously suspected." In our characterization results, we not only show that rational aggregate choices may be individually fully irrational but also that the aggregation of individually irrational behaviors may well be stochastically rational and that individuals with correlated mistakes can never be translated into rational aggregate behavior. These observations resonate within the framework of standard economic theory where market efficiency is built upon uncorrelated errors of individual investors.

\newpage

\appendix

\section{Auxiliary results} \label{app: aux}

\subsection{Preliminaries on convex analysis}

A set $X \subset \mathbb{R}^d$ is convex if for all $x,y \in X$ and $\lambda$ $\in (0,1)$ it holds that $\lambda x +(1-\lambda) y$ $\in X$. The convex hull of a set $Y \subseteq \mathbb{R}^d$ is the intersection of all convex sets containing $Y$ and is denoted $\text{conv}(Y)$, or equivalently the set of all points that can be obtained as convex combinations of points from $Y$. A point $x \in X \subseteq \mathbb{R}^d$ is called an \emph{extreme point} of $X$ if there is no $\lambda$ $\in (0,1)$ and $y,z \in X\setminus \{x\}$ with $y \neq z$ such that $x = \lambda y +(1-\lambda) z$. In other words, a point $x \in X$ is an extreme point if it cannot be written as a non-trivial convex combination of two other distinct points in $X$. Denote the set of extreme points of a convex set $X$ by $\text{ext}(X)$. The next result is a basic result from convex analysis \cite*[p.155]{Rock70} and will be employed when showing the sufficiency of the correlation bounds in Theorem \ref{thm: RNUM}.

\begin{lemma}\emph{(Krein--Milman theorem in finite dimension).} \label{lemma: cat} Every point $x$ in a convex set $X \subseteq \mathbb{R}^d$ can be written as a convex combination of the extreme points of $X$. I.e. $\text{conv}(\text{ext}(X)) = X$.\end{lemma}

\subsection{Preliminaries on stochastic rationality and Block-Marschak polynomials} \label{app: BM}
\

The well-known characterization of the RUM is based on the Block-Marschak-polynomials (BM-polynomials). Define for all $A \subseteq X$ and $a \in X$ the BM-polynomial at $a$ in $A$ by: $$\text{BM}_{\rho}(a,A) = \sum_{A \subseteq B \subseteq X} (-1)^{|B\setminus A|} \rho(a,B)$$

\begin{theorem}[\cite{block}, \cite{falmagne1978representation}, \cite{barbera1986falmagne}, \cite{monderer1992stochastic}, \cite{fiorini2004short}]
$\rho$ is a RUM if and only if $\text{BM}_{\rho}(a,A) \geq 0$ for all $A \subseteq X$ and $a \in X$.
\end{theorem}

We now state two facts about BM-polynomials. The first shows that the BM polynomials act as a convex operator on the set of stochastic choice functions. 

\begin{lemma} \label{lemma: convex} Let $\alpha \in (0,1)$ and $\rho,\rho'$ be stochastic choice functions then $$\text{BM}_{\alpha \rho + (1-\alpha) \rho'}(a,A)  = \alpha \text{BM}_{\rho}(a,A) + (1-\alpha) \text{BM}_{\rho'}(a,A)$$ for all $a \in A$ and $A \subseteq X$. \end{lemma}

\begin{proof} The proof is simple and follows by expanding the definition of a BM polynomial. We have that $$\text{BM}_{\alpha \rho + (1-\alpha) \rho'}(a,A) = \sum_{A \subseteq B \subseteq X} (-1)^{|B\setminus A|} \left[ \alpha \rho(a,B) + (1-\alpha) \rho'(a,B) \right] = $$ $$ \alpha \sum_{A \subseteq B \subseteq X} (-1)^{|B\setminus A|}  \rho(a,B) + (1-\alpha) \sum_{A \subseteq B \subseteq X} (-1)^{|B\setminus A|}  \rho'(a,B)= $$ $$ \alpha \text{BM}_{\rho}(a,A) + (1-\alpha) \text{BM}_{\rho'}(a,A).$$ \end{proof}

The second shows that RUMs with full support are characterized by strictly positive BM-polynomials.

\begin{lemma}  A RUM is full support if and only if $\text{BM}_{\rho_{\mu}}(a,A)>0$ for all $a \in A$ and $A \subseteq X$. \end{lemma}

\begin{proof} This result follows from the discussion in \cite[p.7]{turansick}  \end{proof}

\subsection{Preliminaries on Fréchet bounds} \label{app: frechet}

The Fr\'echet bounds govern the relationship between joint and marginal distributions. Specifically, they come as bounds on the joint distribution given the marginals. In a general setting, the Fr\'echet  bounds are obtained as follows.

\begin{lemma} \label{lemma: frechet}  Let $\Omega$ be a finite sample space and $\pi$ a probability measure on $\Omega$. For any collection of events $E_1,...,E_{m} \subseteq \Omega$ it then holds that $$\max\{0,  \sum_{j=1}^{m} \pi(E_{j}) - (m-1)\} \leq \pi(\cap_{j=1}^{m} E_{j}).$$  \end{lemma}

The Fréchet bounds follow from the basic laws of probability and an induction argument.\footnote{More specifically, they follow by noting that for any two events $E_1, E_2$ it holds that $\pi(E_1 \cap E_2) = \pi(E_1) + \pi(E_2) - \pi (E_1 \cup E_2)$ which implies that  $\pi(E_1 \cap E_2) = \pi(E_1) + \pi(E_2) - \pi (E_1 \cup E_2) \geq \pi(E_1) + \pi(E_2) - 1$ and since also $\pi(E_1 \cap E_2) \geq 0$ we have $\pi(E_1 \cap E_2) \geq \max\{0, \pi(E_1) + \pi(E_2) - 1\}$. An induction argument then completes the proof.} In RCMs, $\mu$ is the joint distribution while $\rho_{\mu}$ is the marginal distribution, and the operation $\rho(c(A),A) =\mu(\{c' \in \mathcal{C}: c'(A) = c(A)\})$ is nothing more than a marginalization. The following result is hence a direct corollary to Lemma \ref{lemma: frechet}.  

\begin{lemma} \label{lemma: strict1} Let $\rho$ be a stochastic choice function, $c$ be a choice function, and $\mathcal{B}$ be an arbitrary collection of choice sets. It then holds that $$ \max\{0,  \sum_{A \in \mathcal{B}} \rho(c(A),A) - (|\mathcal{B}|-1)\} \leq  \mu( \cap_{A \in \mathcal{B} } \{c' \in \mathcal{C}: c'(A)=c(A) \} )$$ for all RCMs $\mu$ such that $\rho= \rho_{\mu}$. \end{lemma}

It should be noted that the lefthand side of the above inequality is independent of the RCM $\mu$ that represents $\rho$. For our purposes, it is also instructive to consider the special case of Lemma \ref{lemma: strict1} when $\mathcal{B}$ is the collection of all non-empty and non-singleton menus of $X$, i.e. when $\mathcal{B} =\mathcal{A}$ and $K = \vert \mathcal{B} \vert$. Then, since $\cap_{A \in \mathcal{B} } \{c' \in \mathcal{C}: c'(A)=c(A) \} = \{c\}$ we have the following result.  

\begin{lemma}\label{lemma: full}  Let $c$ be a choice function then it holds that $$\max\{0,  \sum_{A \in \mathcal{B}} \rho(c(A),A) - (K-1)\} \leq  \mu(c)$$  for all RCMs $\mu$ such that $\rho= \rho_{\mu}$.\end{lemma}

By this formulation of the Fréchet bounds it is clear that if the left-hand side of the inequality is strictly positive, then any $\mu$ that represents $\rho$ must assign (strictly) positive probability to the choice function on the right-hand side of the inequality. In particular, if $c$ is rationalizable then any RCM $\mu$ must assign positive probability to $c$, implying that $\rho_{\mu}$ is not an I-RUM. Equivalently, if a RUM has an irrational representation, so that $\mu(P)=0$ for all preferences $P$, then Lemma \ref{lemma: full} implies the following bounds.

\vspace{5mm}

\noindent \textbf{Weak correlation bounds.} A stochastic choice function $\rho$ satisfies the \emph{weak correlation bounds} if for all preferences $P$ it holds that $$ \sum_{A \in \mathcal{A}} \rho(c_P(A),A) - [K-1]  \leq 0.$$ 

Although Lemma \ref{lemma: full} implies that the weak correlation bounds are necessary for an irrational representation, they are not quite sufficient. The following example illustrates. 

\vspace{5mm}

\begin{innercustomthm}
Let $X=abc$ and let $P_1,P_2,P_3,P_4$ be preferences on $X$ s.t. $a P_1 b P_1  c$, $a P_2  c P_2  b$ and $b P_3  c P_3  a$, $c P_4  b P_4  a$. Let $\mu$ be a RUM with  $\mu(P_1) = \mu(P_2) = 0.4$ and $\mu(P_3) = \mu(P_4) = 0.1$. The following $\rho_\mu$ arises: $\rho_{\mu}(a,X) = \rho_{\mu}(a,ab) = \rho_{\mu}(a,ac) = 0.8$, $\rho_{\mu}(b,bc) = 0.5$. The two choice functions $c_{P_1}, c_{P_2}$ related to $P_1, P_2$ are salient and they both satisfy the weak correlation bounds 
$$\sum_{A \in \mathcal{A}} \rho(c_{P_1}(A),A) = \sum_{A \in \mathcal{A}} \rho(c_{P_2}(A),A) = 2.9 < 3 = K - 1.$$
However, one can easily see that if $\mu'$ is an RCM with $\mu'(c_{P_1}) = \mu'(c_{P_2}) = 0$ then $\rho$ cannot be represented by $\mu'$ (i.e. $\rho \neq \rho_{\mu'}$).
\end{innercustomthm}

The problem in the above example is that $P_1$ and $P_2$ are close to being the same preference since they agree on the choices from all menus except  $\{b,c\}$ (which contains the worst two alternatives of both orders). Moreover, a large fraction of the overall probability is assigned to $P_1$ and $P_2$ (i.e. probability $0.8$) and hence the stochastic choice function $\rho_{\mu}$ exhibits a high degree of correlation. It turns out that a slight variation of these bounds, namely the correlation bounds of section \ref{sec: RNUM}, are sufficient for an irrational representation of RUM.  

\vspace{5mm}

\noindent \textbf{Relation to Kendall rank correlation.} Before providing the proof of Theorem \ref{thm: RNUM}, we would like to briefly motivate the use of the word "correlation" which intends a measure of concordance between the choices of the decision-makers in the population w.r.t. to a specific choice function. To see this, note that, w.r.t to a choice function $c$, we can split the probability mass in $\rho$ between the one that "correlates" and "uncorrelates" with $c$.
$$\bigg[ \underbrace{\sum_{A \in \mathcal{A}} \rho(c(A), A)}_\text{"correlation" $(\hat{\mathbb{C}}^{\rho}_{c})$} + \underbrace{\sum_{A \in \mathcal{A}} \rho(A \setminus c(A), A)}_\text{"uncorrelation" $(\hat{\mathbb{U}}^{\rho}_{c})$} \bigg] = K$$
Finally, we note that $\frac{1}{  K}\hat{\mathbb{C}}^{\rho}_{c}, \frac{1}{  K}\hat{\mathbb{U}}^{\rho}_{c} \in [0,1]$ and the equation above implies $\frac{1}{  K}[\hat{\mathbb{C}}^{\rho}_{c} - \hat{\mathbb{U}}^{\rho}_{c}] \in [-1, 1]$. An intuitive "correlation" interpretation comes after noticing the relationship with the notion of Kendall rank correlation as $\hat{\mathbb{C}}^{\rho}_{c}$ is the probability mass on the choices that are concordant with $c$ while $\hat{\mathbb{U}}^{\rho}_{c}$ is the probability mass on the choices that are discordant with $c$, and $\frac{1}{K}$ is the normalization factor.

\section{Proof of Theorem \ref{thm: RNUM}}  \label{app: main} 

\subsection{The necessity of the correlation bounds} \label{sec: nproof}
The necessity of the correlation bounds follows from an application of the Fréchet bounds in Lemma \ref{lemma: strict1}. Recall that $\mathcal{A}(P)= \mathcal{A} \setminus \{\{x_{P}(N-1),x_{P}(N)\}\}$ and $|\mathcal{A}(P)| = K-1$ for all preferences $P$. The following lemma is a special case of Lemma \ref{lemma: strict1} and it follows by noting that for each preference $P$ we have $\cap_{A \in \mathcal{A}(P)} \{c \in \mathcal{C}: c(A)=c_P(A)\} = \{c_P,c_{Q}\}$ where $Q$ agrees with $P$ on all the menus in $\mathcal{A}(P)$.

\begin{lemma} \label{lemma: strict} Let $\rho$ be a stochastic choice function, $P$ be a preference, and $Q$ be a preference that agrees with $P$ on all menus in $\mathcal{A}(P)$. It then holds that, for all RCMs $\mu$ such that $\rho= \rho_{\mu}$, $$\sum_{A \in \mathcal{A}(P)} \rho(c_P(A),A) - (K-2) \leq  \mu( \{c_P,c_Q\} ).$$ \end{lemma}

Since any RNUM $\mu$ puts zero probability on all the rational choice functions, i.e. $\mu(c_P)=0$ for all preferences $P$, Lemma \ref{lemma: strict} implies that any RNUM must satisfy the following inequalities for all preferences $P$:

$$\mathbb{C}^{\rho}_{P} - [K-2] = \sum_{A \in \mathcal{A}(P)} \rho(c_{P}(A),A) - [K-2]  \leq 0,\quad\text{}$$

which are exactly the correlation bounds. 

\subsection{The sufficiency of the correlation bounds}

We begin with one simple observation. Because the correlation bounds are linear inequalities, the set of stochastic choice functions that satisfy them is convex.

\begin{lemma} \label{lemma: convB} The set of stochastic choice functions that satisfy the correlation bounds is a convex set. \end{lemma}

The proof of Lemma \ref{lemma: convB} is straightforward and therefore omitted. Given Lemma \ref{lemma: convB}, by a standard result in convex analysis (Krein--Milman theorem, see Lemma \ref{lemma: cat}), any point in this set can be expressed as a convex combination of its extreme points. It therefore suffices to show that each extreme point admits an RNUM representation. We prove this through three lemmas.

\begin{lemma} \label{lemma: dual} Let $\rho$ be a dual RCM with distribution $\mu$ that satisfies the correlation bounds with equality at $P$ with $\mu(P)>0$, then $\rho$ is an RNUM. \end{lemma}

\noindent\emph{Idea of the proof.} The argument underlying Lemma \ref{lemma: dual} is to show that a dual RCM $\mu$ that satisfies the correlation bounds with equality at $P$ yields the same stochastic choice function as an RNUM that splits the rational mass of $\mu$ over a family of choice functions, each of which deviates from $c_P$ in exactly one menu, and each of which is irrational. The number of such deviations available is exactly the number of menus on which the two choice functions in the support of $\mu$ disagree, and the rational mass can be spread evenly over these choice functions. In particular, if $\mu$ is a dual RUM, its RNUM representation has support on 1-step away choice functions. Accordingly, the proof is divided into three steps:
\begin{itemize}
    \item[Step 1.] The construction of these choice functions,
    \item[Step 2.] The verification that each of them is irrational, and
    \item[Step 3.] The verification that the uniform distribution over them reproduces $\rho$.
\end{itemize}

\begin{proof} Step 1. Let $\rho_{\mu}$ be a dual RCM such that $\mu(P)> 0$ and $\mathbb{C}^{\rho}_P - [K-2]=0$ for some preference $P$. Thus, $\rho$ has support $\{c,P\}$ where $c \in \mathcal{C} \setminus \{P\}$. Since $\mathbb{C}^{\rho}_P - [K-2]=0$ it follows that $\mu(c)n(P,c) + \mu(P)n(P,P)=K-2$. Rearranging the latter equality it follows that $\mu(P) = \frac{K- 2 - n(P,c) }{K -1 -n(P,c)}$. Let $k \in \mathbb{N}$ be  such that $k \mu(c) = \mu(P)$. I.e. $k= K- 2- n(P,c) $. Let the set $\mathcal{D}(P)$ be defined as $\mathcal{D}(P)= \{A\in \mathcal{A}(P) : c_{P}(A) \neq c(A)\}$ (which has cardinality $k+1$). Pick an arbitrary $B_0 \in \mathcal{D}(P)$ and define a choice function $c_{B_0}$ by
$$ c_{B_0}(A) =
\begin{cases}
c(A) & \text{if $A \in \{B_0, \{x_P(N-1),x_P(N)\}\}$} \\
c_{P}(A) & \text{otherwise.} \\
\end{cases}
$$
For each $B \in \mathcal{D}(P)\setminus\{B_0\}$ define a choice function $c_B$ by

$$ c_{B}(A) =
\begin{cases}
c(A) & \text{if $A = B$} \\
c_{P}(A) & \text{if $A  \neq  B$} \\
\end{cases}
$$

\vspace{5mm}

\noindent Step 2. We next show that $c_{B_0}$ is irrational and that $c_B$ is irrational for all $B \in \mathcal{D}(P)\setminus\{B_0\}$. For a preference $P$ and for each $i \in \{1,\dots,N\}$ we define $x_{P}(i)$ as the alternative ranked at position $i$ according to $P$, i.e. $|\{x \in X : x P x_P(i)\}| =i-1$. Let $B \in \mathcal{D}(P)$. Since $B \neq \{x_{P}(N-1),x_{P}(N)\}$, it follows that $x_{P}(j) \in B$ for some $j <N-1$. Let $j^*$ be the smallest such $j$. Note that $x_P(j^*) \in B \subseteq \{x_{P}(j^*),\dots,x_{P}(N)\}$. We split into two cases. First, if $|B| = 2$, then there is a ternary menu $C$ with $B \subset C \subseteq \{x_{P}(j^*),\dots,x_{P}(N)\}$. But since $x_P(j^*) \in C$ and $B \neq C \subseteq \{x_{P}(j^*),\dots,x_{P}(N)\}$ we have that $c_B(C) = c_P(C) = x_P(j^*)$ and $c_B(B) = c(B) \neq c_P(B) = x_P(j^*)$. A violation of Sen's property $\alpha$. If $|B| \geq 3$, then there is a binary menu $C\,=\{x_P(j^*),c(B)\}$ with $C \subset B \subseteq \{x_{P}(j^*),\dots,x_{P}(N)\}$ and $x_P(j^*) \in C$. Since $x_P(j^*) \in C$ and $C \neq B \subset \{x_{P}(j^*),\dots,x_{P}(N)\}$ we have that $c_B(C) = c_P(C) = x_P(j^*)$. Since $c_B(B)=c(B)\in C$ and $c_B(B) \neq x_P(j^*)$, a violation of Sen's property $\alpha$.

\vspace{5mm}

\noindent Step 3. Let $\mu'$ be the uniform distribution on $\{c_A : A \in \mathcal{D}(P)\}$, i.e.$$\mu'(c_A) =\frac{1}{k+1}$$ for all $A \in \mathcal{D}(P)$. To show that $\mu$ is an RNUM, it remains to show that the stochastic choices generated by $\mu$ and $\mu'$ are the same. There are three cases.

\vspace{5mm}

\noindent Case 1: If $A \in \mathcal{A}$ is such that $A \neq \{x_P(N-1),x_P(N)\}$ and $A \neq B$ for all $B \in \mathcal{D}(P)$ then $c_B(A) =c_{P}(A) = c(A)$ for all $B \in \mathcal{D}(P)$ and it is clear that the stochastic choices generated by $\mu$ and $\mu'$ are the same (i.e. that $\rho_{\mu}(\cdot,A)=\rho_{\mu'}(\cdot,A)$).
\vspace{5mm}

\noindent Case 2: If $A= B$ for some $B \in \mathcal{D}(P)$. Then $c_B(A) = c(A)$ and $c_{B'}(A) = c_{P}(A)$ for all $B' \in \mathcal{D}(P) \setminus \{B\}$. It hence follows that

 $$\sum_{B' \in \mathcal{D}(P)} \mu'(c_{B'}) \bold{1}\{a =c_{B'}(A)\} =$$  $$=\frac{1}{k+1}\bold{1}\{a =c_B(A)\}  +  \sum_{B' \in \mathcal{D}(P) \setminus \{B\}} \frac{1}{k+1} \bold{1}\{a =c_{B'}(A)\}  = $$   $$=\frac{1}{k+1} \bold{1}\{a =c_B(A)\} + \frac{k}{k+1} \bold{1}\{a =c_{P}(A)\} $$

 $$ =\mu(c)\bold{1}\{a =c(A)\} + \mu(P) \bold{1}\{a =c_{P}(A)\},$$ where the last equality follows since $k\mu(c) = \mu(P)$ implies that $\mu(c)=\frac{1}{k+1}$ and $\mu(P)=\frac{k}{k+1}$.

\vspace{5mm}

\noindent Case 3: If $A = \{x_P(N-1),x_P(N)\}$ then $c_{B_0}(A) = c(A)$ and $c_B(A) = c_P(A)$ for all $B \in \mathcal{D}(P)\setminus\{B_0\}$. The same calculation as in Case 2 then shows that $\mu$ and $\mu'$ generate the same choices. \end{proof}

\begin{lemma} \label{lemma: correq}
Let $\mu$ be an RCM such that the correlation bounds are satisfied with equality. Then there are dual RCMs $(\hat{\mu_i})_{i=1}^{n}$ such that

\begin{enumerate}

\item \label{lemma: eq1} $\mu$ is a convex combination of the $(\hat{\mu_i})_{i=1}^{n}$, i.e. there are weights $\delta_i$ s.t. $\mu = \sum_{i=1}^n \delta_i \hat{\mu_i}$,

\item \label{lemma: eq2} each dual RCM $\hat{\mu_i}$ satisfies the correlation bounds with equality.

\end{enumerate}

\end{lemma}

\noindent\emph{Idea of the proof.} The second lemma shows that dual RCMs are the keystones of our argument. Other RCMs that satisfy the correlation bounds with equality can all be written as convex combinations of dual RCMs in which one side is always the preference $P$ at which the correlation bounds bind, while the other side is a generic choice function. In particular, if $\mu$ is a RUM, it can be written as a convex combination of dual RUMs that satisfy the correlation bounds with equality.

Suppose that the correlation bound binds at $P$, and let $Q$ be its twin preference. We call an RCM $\mu$ an \emph{almost dual RCM} if $\text{supp}\,\mu \subseteq \{P,Q,c\}$ for some choice function $c \in \mathcal{C}$. The key observation is that the binding correlation bound provides exactly the accounting identity needed to decompose $\mu$ into such almost dual RCMs. In particular, $P$ and $Q$ have maximal correlation with $P$, whereas every other choice function $c$ falls short by $K-1-n(P,c)$. We use these correlation deficits to construct one almost dual RCM for each $c \notin \{P,Q\}$, with the mass on $P$ and $Q$ chosen so that the correlation bound continues to bind. The binding equality for $\mu$ ensures that these almost dual RCMs can be weighted so as to reproduce exactly the original distribution $\mu$.

The substance of the proof is to establish this decomposition into almost dual RCMs. It is then easy to show that each almost dual RCM is itself a convex combination of two dual RCMs, supported respectively on $\{P,c\}$ and $\{Q,c\}$, both of which satisfy the correlation bounds with equality.

\begin{proof} Let $\mu$ be an RCM that satisfies the correlation bounds with equality.  W.l.o.g. assume that the correlation bound is satisfied with equality at $P$ (i.e. we have $ K - 2 = \sum\limits_{c \in \mathcal{C}} n(P, c) \mu(c)$). Let $Q$ be the "twin"-preference of $P$. Note that $n(P,c)=n(Q,c)$  for all $c \in \mathcal{C}$ and that the correlation bound is also satisfied with equality at $Q$. Assume that $\mu(P)+\mu(Q)=0$, but that $\mu(P')>0$ for some preference $P' \notin \{P,Q\}$, then
$$\sum_{c \in \mathcal{C}} n(P,c)\mu(c) = \sum_{c \in \mathcal{C}\setminus\{P,Q,P'\}} n(P,c)\mu(c) + n(P,P')\mu(P') \leq$$
$$\sum_{c \in \mathcal{C}\setminus\{P,Q,P'\}} [K-2]\mu(c) + [K-3]\mu(P') < K-2,$$
and this is a contradiction. Hence, throughout the proof: $\mu(P)+\mu(Q)>0$.

\vspace{5mm}

\noindent We now come to the heart of the proof. We construct a series of almost dual RCMs  $(\hat{\mu}_c)_{c \in \mathcal{C}\setminus\{P,Q\}}$  such that $\mu$ is a convex combination of $(\hat{\mu}_c)_{c \in \mathcal{C}\setminus\{P,Q\}}$ and such that each $\hat{\mu}_c$ satisfies the correlation bounds with equality.

\vspace{5mm}

 \noindent For each $c \in \mathcal{C} \setminus \{P,Q\}$ we define an almost dual RCM $\hat{\mu}_c$ with support on $\{P,Q,c\}$ such that the correlation bounds are satisfied with equality

$$ \hat{\mu}_c(P) =\frac{\mu(P)}{\mu(P)+\mu(Q)} \left[ \frac{K - 2 - n(P, c)}{K -1- n(P,c)} \right]$$
$$ \hat{\mu}_c(Q) =\frac{\mu(Q)}{\mu(P)+\mu(Q)} \left[ \frac{K - 2 - n(P, c)}{K -1- n(P,c)} \right]$$
$$\hat{\mu}_c(c) = \frac{1}{K -1- n(P,c)}.$$

\vspace{5mm}

\noindent Note that the probabilities above are well-defined since $n(P,c) < K-1$ for all $c \in \mathcal{C} \setminus \{P,Q\}$

\vspace{5mm}

\noindent A sequence of weights $(\delta_c)_{c \in \mathcal{C} \setminus \{P,Q\}}$  is then defined for each $c \in \mathcal{C} \setminus \{P,Q\}$ as follows:
 $$ \delta_c =  \mu(c) [K -1- n(P, c)] $$
 \noindent We next show that $\sum_{c \in \mathcal{C} \setminus \{P,Q\}} \delta_c =1$, i.e. that the $\delta_c$ are indeed weights (that $\delta_c \geq 0$ is obvious). We have that
$$\sum_{c \in \mathcal{C} \setminus \{P,Q\}} \delta_c  = \sum_{c \in \mathcal{C} \setminus \{P,Q\}} \mu(c)[K -1- n(P, c)]  = $$
$$[K-1][1- \mu(P)-\mu(Q)] - \sum_{c \in \mathcal{C} \setminus \{P,Q\}}  n(P,c) \mu(c) =$$
$$1 + K-2 - \sum_{c \in \mathcal{C} } n(P, c) \mu(c) = 1,$$ where the first equality follows by definition of  $\delta_c$, the second equality follows since $\sum_{c \in \mathcal{C} \setminus \{P,Q\}} \mu(c)[K -1]=[K-1][1- \mu(P)-\mu(Q)]$, the third equality follows by noting that $n(P, P)\mu(P)+n(P,Q)\mu(Q) =[K-1] [\mu(P)+\mu(Q)]$ and rearranging, and the final equality follows since the correlation bound holds with equality at $P$.

\vspace{5mm}

\noindent It remains to show that $\mu$ is a convex combination of the almost dual RCMs  $(\hat{\mu}_c)_{ c \in \mathcal{C}\setminus \{P,Q\}}$. It is clear that for all $c \in \mathcal{C} \setminus \{P,Q\}$ it holds that
$$\mu(c) = \delta_c \hat{\mu}_c(c) =\sum_{c' \in \mathcal{C} \setminus \{P,Q\}} \delta_{c'}\hat{\mu}_{c'}(c),$$ where the first equality follows by construction of $\delta_c$ and $\hat{\mu}_c$, and the second equality follows since $\hat{\mu}_{c'}(c)= 0$ for all $c' \notin \{c,P,Q\}$. We next note that
$$\sum_{c \in \mathcal{C} \setminus \{P,Q\}} \delta_c \hat{\mu}_c(P)=$$
$$= \frac{\mu(P)}{\mu(P)+\mu(Q)} \sum_{c \in \mathcal{C} \setminus \{P,Q\}} \mu(c) [K -2- n(P, c)] = $$
$$ \frac{\mu(P)}{\mu(P)+\mu(Q)} \left[\sum_{c \in \mathcal{C} \setminus \{P,Q\}} \mu(c) [K -1- n(P, c)] + \mu(P) + \mu(Q)-1 \right] $$
$$= \frac{\mu(P)}{\mu(P)+\mu(Q)} [\mu(P)+\mu(Q)]= \mu(P),$$ and similarly it follows that $\mu(Q)=\sum_{c \in \mathcal{C} \setminus \{P,Q\}} \delta_c \hat{\mu}_c(Q)$.

\vspace{5mm}

\noindent To conclude, it remains only to split each almost dual RCM into dual RCMs. Suppose that $\mu$ is an almost dual RCM that satisfies the correlation bound with equality at $P$ and that $\mu$ has support $\{P,Q,c\}$ (where $Q$ is the "twin"-preference of $P$).
We may then define two dual RCMs $\mu_1$ and $\mu_2$ by
$$\mu_1(P) =\mu_2(Q)= \frac{K - 2 - n(P, c)}{K -1- n(P,c)}$$
$$\mu_1(c)=\mu_2(c) = \frac{1}{K -1- n(P,c)}.$$

\noindent Both $\mu_1,\mu_2$ satisfy the correlation bounds with equality. Further, define weights $\delta_1 = \frac{\mu(P)}{\mu(P)+\mu(Q)}  \mu(c) [K -1- n(P, c)] $ and $\delta_2 =  \frac{\mu(Q)}{\mu(P)+\mu(Q)}  \mu(c) [K -1- n(P, c)]$. It is straightforward to verify that $\delta_1,\delta_2$ are indeed weights and that $\mu = \delta_1 \mu_1+ \delta_2\mu_2$, which completes the proof. \end{proof}

\begin{lemma}  \label{lemma: extreme} Let $\rho$ be an RCM with distribution $\mu$ (so $\rho=\rho_{\mu}$) that is an extreme point of the set of stochastic choice functions that satisfy the correlation bounds. Then $\rho$ is either an RNUM or $\rho$ satisfies the correlation bounds with equality.  \end{lemma}

\noindent\emph{Idea of the proof.} We argue by contrapositive. Suppose that $\mu$ assigns positive probability to some preference $P_0$, so that $\mu$ is not an RNUM, and that all correlation bounds hold with strict inequality. We show that $\rho_{\mu}$ cannot be an extreme point. The idea is to use some other choice function $c_1$ in the support of $\mu$ to construct a dual RCM supported on $\{P_0,c_1\}$ that lies on the boundary of the set defined by the correlation bounds. We then remove a sufficiently small amount of this dual RCM from $\mu$. Since all correlation bounds have slack at $\rho_{\mu}$, this can be done while leaving a residual RCM that still satisfies all the correlation bounds. Reversing the construction expresses the original RCM as a non-trivial convex combination of the dual RCM and the residual RCM. Since both satisfy the correlation bounds, $\rho_{\mu}$ cannot be an extreme point.

\begin{proof} If $\rho$ is an RNUM the proof is complete. Thus, assume that $\rho$ is an RCM with distribution $\mu$ with support on at least one preference $P_0$. We show that if $\rho=\rho_{\mu}$ satisfies the correlation bounds with strict inequality (i.e. $\mathbb{C}^{\rho}_{P} - [K-2]<0$ for all preferences $P$) then $\rho$ is not an extreme point. By assumption, the support of $\mu$ is not contained in $\{c_{P_0},c_{Q_0}\}$, where $Q_0$ is the twin of $P_0$ (otherwise $\mathbb{C}^{\rho}_{P_0}=K-1$). Then there is a choice function $c_1 \in \mathcal{C} \setminus \{P_0,Q_0\}$ such that $\{c_{P_0},c_1\}$ is contained in the support of $\mu$. 
Let $\alpha$ be a distribution with support on $c_{P_0}, c_1$ such that the dual RCM $\rho^d$, defined for all $a \in A$ and $A \in \mathcal{A}$ by $$\rho^d(a,A) = \alpha(c_{P_0}) \bold{1}\{a= c_{P_0}(A)\} + \alpha(c_1) \bold{1}\{a= c_1(A)\},$$ satisfies the correlation bounds with equality at $P_0$ (hence at all preferences, as $\alpha(c_1)\leq\frac{1}{2}$ unless $\rho^d$ is concentrated on the irrational $c_1$). For each $\beta \in [0,1)$ such that $\mu(c) - \beta \alpha(c) > 0$ for all $c \in \{c_{P_0},c_1\}$  (or equivalently $\beta < \min\{\frac{\mu(c_{P_0})}{\alpha(c_{P_0})},  \frac{\mu(c_1)}{\alpha(c_1)}\}$) define an RCM $\mu^*_{\beta}$  by

 $$\mu^*_{\beta}(c) =
  \begin{cases}
  \frac{\mu(c) - \beta \alpha(c) }{1-\beta} & \text{if  $c \in \{c_{P_0},c_1\}$} \\
 \frac{\mu(c)}{1 - \beta} & \text{otherwise.}
  \end{cases}
  $$

 \vspace{5mm}

\noindent Let $\rho^*(\beta) = \rho_{\mu^*_{\beta}}$, i.e. $\rho^*(\beta)$ is the RCM stochastic choice function with distribution $\mu^*_{\beta}$. For each preference $P$ define a function $f_P$ by $$f_P(\beta) = \mathbb{C}^{\rho^*(\beta)}_{P}$$ for all $\beta \in \left[0, \min \left\{\frac{\mu(c_{P_0})}{\alpha(c_{P_0})},  \frac{\mu(c_1)}{\alpha(c_1)}\right\} \right)$. First note that $f_P(0)=\mathbb{C}^{\rho^*(0)}_{P}=\mathbb{C}^{\rho}_{P}$ for all preferences $P$. Hence, since the correlation bounds are satisfied with strict inequality we have $f_P(0) < K-2$ for all preferences $P$. Define a (vector-valued) function $f(\beta): =(f_P(\beta))_{P \in \mathcal{P}}$. Since each $f_P$ is a continuous function of $\beta$ it follows that $f$ is a continuous function of $\beta$ and hence there is a $\hat{\beta} > 0$ such that $f_P(\beta) < K-2$ for all $0 \leq \beta < \hat{\beta}$ and for all preferences $P$. Let $\beta$  be such that $0<\beta < \hat{\beta} <\min \left\{\frac{\mu(c_{P_0})}{\alpha(c_{P_0})},  \frac{\mu(c_1)}{\alpha(c_1)}\right\}$. We then have that $\rho = \beta \rho^d + (1-\beta) \rho^*(\beta)$ and since both $\rho^d$ and $\rho^*(\beta)$ satisfy the correlation bounds, it follows that $\rho$ is not an extreme point (as it can be written as a non-trivial convex combination of two other RCMs that satisfy the correlation bounds).  \end{proof}


\noindent\textbf{Proof of sufficiency.} Recall from the discussion following Lemma \ref{lemma: convB} that it suffices to show that every extreme point of the set of stochastic choice functions satisfying the correlation bounds admits an RNUM representation. By Lemma \ref{lemma: extreme}, every extreme point is either an RNUM, in which case we are done, or satisfies the correlation bounds with equality at some preference $P$. In the latter case, Lemma \ref{lemma: correq} shows that it can be expressed as a convex combination of dual RCMs satisfying the correlation bounds with equality at $P$. Since each such dual RCM admits an RNUM representation by Lemma \ref{lemma: dual}, and a mixture of RNUMs is itself an RNUM, the result follows.

\section{Other proofs omitted from main text}   \label{app: other}

\subsection{Proof of Proposition \ref{pro: uniqueness}}

Assume that $\mathbb{C}^\rho_P - [K-2] = 0$. We first show that any RNUM representation $\mu$ of $\rho$ has support contained in the set of 1-step away choice functions.

Let $\mu$ be an RNUM representation of $\rho$. Suppose by contradiction that $\mu(c^*) > 0$ for some irrational choice function $c^*$ that is not 1-step away from $P$, so that $n(P,c^*) \leq K-3$. Define
\[
\rho_{c^*}(a, A) = \textbf{1}\{c^*(A) = a\} \quad \text{and} \quad \hat{\rho}(a, A) = \frac{\rho(a, A) - \mu(c^*)\textbf{1}\{c^*(A) = a\}}{1 - \mu(c^*)}
\]
for all $a \in A$ and $A \subseteq X$. Then it is clear that $\rho$ is a convex combination of $\hat{\rho}$ and $c^*$, i.e.,
\[
\rho(a, A) = \mu(c^*)\textbf{1}\{c^*(A) = a\} + (1 - \mu(c^*))\hat{\rho}(a, A) \quad \forall a \in A, A \subseteq X.
\]

Next, note that $\mathbb{C}^{\rho_{c^*}}_P - [K-2] \leq [K - 3] - [K-2] < 0$. Further, since $\mathbb{C}^\rho_P - [K-2] = 0$ and since
\[
\mathbb{C}^\rho_P - [K-2] = \mu(c^*) \big(\mathbb{C}^{\rho_{c^*}}_P - [K-2]\big) + (1 - \mu(c^*)) \big(\mathbb{C}^{\hat{\rho}}_P - [K-2]\big),
\]
we must have $\mathbb{C}^{\hat{\rho}}_P - [K-2] > 0$. But if $\mathbb{C}^{\hat{\rho}}_P - [K-2] > 0$, then any RCM representation of $\hat{\rho}$ puts positive probability on $P$. Since the RCM $\hat{\mu}$ defined by
\[
\hat{\mu}(c) = \frac{\mu(c)}{1 - \mu(c^*)} \quad \text{for all } c \in \mathcal{C} \setminus \{c^*\}, \quad \hat{\mu}(c^*) = 0
\]
represents $\hat{\rho}$, it follows that $\hat{\mu}(c_P) > 0$. Thus, it follows that $\mu(c_P) > 0$, a contradiction.

Next, we show that $\rho$ has an essentially unique RNUM representation. Let $\mu, \mu'$ be RNUM representations of $\rho$. By the above, both $\mu$ and $\mu'$ have supports contained in the set of choice functions that are 1-step away from $P$. Let $c \in \mathcal{C}$ with $\mu([c]_P) > 0$. Recall that there is a pair of choice functions $c^1,c^2$ s.t. $[c]_P=\{c^1,c^2\}$ where $c^1(A)=c^2(A)$ for all $A \neq \{x_P{(N-1)}, x_P{(N)}\}$. Since $c^1,c^2$ are 1-step away from $P$ we have that $c^1(B^*)=c^2(B^*) \neq \max(B^*, P)$ for exactly one subset $B^*$ of $X$ with $B^* \neq \{x_P{(N-1)}, x_P{(N)}\}$. Set $b^* =c^1(B^*)=c^2(B^*)$. Note that $c'(B^*) \neq b^*$ for all $c' \in \mathcal{C} \setminus \{c^1,c^2\}$ with $\mu(c') > 0$. Hence, $\rho(b^*, B^*) = \mu([c]_P) > 0$. Since $\rho(b^*, B^*) > 0$ it follows that there is a 1-step away choice function $c'$ with $\mu'(c') > 0$ and $c'(B^*) = b^*$, so we must have $c' \in [c]_P$. Hence $\mu'([c]_P) > 0$, and similar reasoning as before implies that $\mu'([c]_P) = \rho(b^*, B^*)$ and so $\mu'([c]_P) = \mu([c]_P)$.


\subsection{Proof of Proposition \ref{pro: falsifiability}}

\begin{proof} Let $\mathtt{SCF}$ be the collection of all stochastic choice functions. Define a function $F: \mathtt{SCF} \to \mathbb{R}$ by $$F(\rho)= \min_{a \in A: A \subseteq X} \text{BM}_{\rho}(a,A).$$ Viewing $\mathtt{SCF}$ as a convex subset of $\mathbb{R}^d$ (for suitable $d$) it is straightforward to check that $F$ is a continuous function (to see this, first note that $\text{BM}_{\rho}(a,A)$ is continuous for fixed $a \in A$ and $A \subseteq X$, and since the minimum of a finite number of continuous functions is continuous it follows that $F$ is continuous). Let $\mathtt{M}^*$ denote the closure of $\mathtt{M}$. Since $\mathtt{M}^*$ is compact and $F$ is continuous there is a $M > 0$ such that $F(\rho) \geq - M$ for all $\rho \in \mathtt{M}^*$ (hence for all $\rho \in \mathtt{M}$). Since $\rho_{\mu}$ is a full support RUM there is an $\beta > 0$ such that  $\text{BM}_{\rho_{\mu}}(a,A) \geq \beta$ for all $a \in A$ and $A \subseteq X$.   Let $1> \bar{\alpha} >0$ be large enough  such that $\bar{\alpha} > \frac{M}{\beta + M} > 0$ (such a number clearly exists). 

\vspace{5mm}

\noindent Let $\alpha \geq \bar{\alpha}$ and $\rho \in \mathtt{M}$. We need to prove that $\rho^* = \alpha \rho_{\mu} + (1-  \alpha) \rho$ is a RUM. It suffices to show that the BM-polynomials of $\rho^*$ are non-negative. Let $a \in A$ and $A \subseteq X$. There are two cases. 

\vspace{5mm}

\noindent Case 1. Assume first that $\text{BM}_{\rho}(a,A) \geq 0$. By Lemma \ref{lemma: convex} we then have $$\text{BM}_{\rho^*}(a,A) = \alpha \text{BM}_{\rho_{\mu}}(a,A) + (1-\alpha) \text{BM}_{\rho}(a,A) \geq 0$$

\vspace{5mm}

\noindent Case 2. Assume next that $\text{BM}_{\rho}(a,A) < 0$. Then $$\text{BM}_{\rho^*}(a,A) = \alpha \text{BM}_{\rho_{\mu}}(a,A) + (1-\alpha) \text{BM}_{\rho}(a,A) \geq $$ $$  \bar{\alpha} \text{BM}_{\rho_{\mu}}(a,A) + (1-\bar{\alpha}) \text{BM}_{\rho}(a,A)  \geq $$ $$\bar{\alpha} \beta + (1-\bar{\alpha}) (-M)  = \bar{\alpha} ( \beta + M) - M >0.$$ The first (in)equality follows by Lemma \ref{lemma: convex} (i.e. convexity of the $\text{BM}$ operator).  The second (in)equality follows since $\alpha \geq \bar{\alpha}$ and $\text{BM}_{\rho_{\mu}}(a,A) > \text{BM}_{\rho}(a,A)$ and the third (in)equality follows since $\text{BM}_{\rho_{\mu}}(a,A) \geq \beta$ and since $\text{BM}_{\rho}(a,A) \geq F(\rho) \geq - M$.  \end{proof}

\subsection{Proof of proposition \ref{pro: almost}} 

\begin{proof} For each preference $P \in \mathcal{P}$ define a  dual RUM $\mu^d_P$ with support on $\{P,P^*\}$, where $P^*$ is the converse of $P$,\footnote{For each preference $P$ define the converse of $P$, denoted $P^*$, for all $a,b \in X$ by $a P ^* b$ if and only if $b P a$.} and such that the correlation bound at $P$ is binding, i.e. so that

$$\mu^d_P(P) = \frac{K- 2 - n(P,P^*) }{K -1 -n(P,P^*)}=\frac{K-2}{K-1} \quad \left(\text{hence} \quad\mu^d_P(P^*)=\frac{1}{K-1}\right).$$ Let $\rho_{P}$ be the degenerate RUM that puts all mass on $P$. Then  $$\lVert \rho_{P}-\rho_{\mu^d_P} \rVert = \max_{a \in A} \max_{A \subseteq X: |A| \geq 2} \left|\rho_{P}(a,A)-\rho_{\mu^d_P}(a,A)\right|  = \mu^d_P(P^*) = \frac{1}{K-1}.$$ Let $\rho$ be RUM. Then there is a distribution $\mu$ such that $\rho=\rho_{\mu}$. Define $$\rho'= \sum_{P \in \mathcal{P}} \mu(P)\rho_{\mu^d_P}.$$   Then, by a final application of the triangle-inequality, it follows that  $$\left\lVert \rho - \rho' \right\rVert   =      \left\lVert \sum_{P \in \mathcal{P}}\mu(P)\left[\rho_P  - \rho_{\mu^d_P}\right] \right\rVert \leq $$ $$ \sum_{P \in \mathcal{P}}\mu(P) \left\lVert  \rho_P  - \rho_{\mu^d_P}\right\rVert    \leq \frac{1}{K-1}=\frac{1}{2^{|X|}-|X|-2}.$$ \end{proof}

\section{Generalization} \label{app: online}

This Appendix extends the main result to incomplete menu collections.

\subsection{A sufficient condition on the domain.}

We introduce additional concepts. A \emph{choice correspondence} is a function $C: \mathcal{A} \to 2^{X} \setminus {\emptyset}$ such that $C(A) \subseteq A$ for all $A \in \mathcal{A}$. We say that $C$ is \emph{maximally rational} if, for every choice function $c \in \mathcal{C}$,

\begin{enumerate}[label=(\roman*)]
\item if $|\{A \in \mathcal{A} : c(A) \notin C(A)\}| = 0$, then $c$ is rational, and
\item if $|\{A \in \mathcal{A} : c(A) \notin C(A)\}| = 1$, then $c$ is irrational.
\end{enumerate}

\begin{definition} \label{def: richness} We call a menu collection $\mathcal{A}$ \emph{rich} (or \emph{a rich menu collection}) if for each preference $P \in \mathcal{P}$ there is a maximally rational choice correspondence $C$ s.t. $c_P(A) \in C(A)$ for all $A \in \mathcal{A}$. \end{definition}

Why is richness of $\mathcal{A}$ sufficient for our main results? The answer lies in the construction of irrational representations of dual RCMs in Lemma \ref{lemma: dual}. In particular, that proof relies on constructing irrational choice functions that are one step away from rational ones. Definition \ref{def: richness} ensures that the domain is sufficiently rich for this construction to go through.

We will later show that a large class of menu collections are indeed rich. As the next example shows, not all menu collections are rich. 

\vspace{5mm} 

\noindent \textbf{Example 1.} Consider  $\mathcal{A}=\{\{a,d\},\{b,c\},\{a,b,c,d\}\},$ and the preference $P=aPbPcPd$. Writing choice functions as $ (c(\{a,d\}),\,c(\{b,c\}),\,c(\{a,b,c,d\})), $ we have $c_P=(a,b,a)$. Suppose that there exists a maximally rational choice correspondence $C$ containing $c_P$. Since $(a,c,b)$ is irrational, clause (i) implies that it cannot be contained in $C$. Hence, either $ c\notin C(\{b,c\})$ or $ b\notin C(\{a,b,c,d\})$. If $c\notin C(\{b,c\})$, then the choice function $(a,c,a)$ differs from $C$ at exactly one menu and is rational, contradicting clause (ii). If $b\notin C(\{a,b,c,d\})$, then $(a,b,b)$ differs from $C$ at exactly one menu and is rational, again contradicting clause (ii). Thus, no maximally rational choice correspondence contains $c_P$, and $\mathcal A$ is not rich.

\vspace{5mm}

For each maximally rational choice correspondence $C$, define $ \mathcal{P}_C = \{P \in \mathcal{P} : c_P(A) \in C(A) \text{ for all } A \in \mathcal{A}\}.$ The next lemma shows that if a menu collection is rich, then the collection of sets $\{\mathcal{P}_C : C \text{ maximally rational}\}$ is a partition of $\mathcal{P}$.

\begin{lemma} Let $\mathcal{A}$ be a rich menu collection. Then the collection of sets $\{\mathcal{P}_C : C \text{ maximally rational}\}$ forms a partition of $\mathcal{P}$. \end{lemma}

\begin{proof} Suppose that $C$ and $C'$ are two maximally rational choice correspondences
with $ \mathcal P_C \cap \mathcal P_{C'} \neq \emptyset.$ Let $P \in \mathcal P_C \cap \mathcal P_{C'}$. Then $c_P(A)\in C(A)\cap C'(A)$
for all $A\in\mathcal A$. We claim that $C=C'$. Suppose not, then, without loss of generality, there
exist $A^*\in\mathcal A$ and $x\in C(A^*)\setminus C'(A^*)$. Define a choice function $c$ by $c(A^*)=x$
and $c(A)=c_P(A)$ for all $A\neq A^*$. Then $c(A)\in C(A)$ for all $A\in\mathcal A$, so by maximal rationality of $C$, condition (i), $c$ is rational. On the other hand, since $x\notin C'(A^*)$ and $c_P(A)\in C'(A)$ for all $A \in \mathcal{A}$, we have $ |\{A\in\mathcal A:c(A)\notin C'(A)\}|=1.$ Therefore, by maximal rationality of $C'$, condition (ii), $c$ is irrational. This is a contradiction. Hence $C=C'$. \end{proof}

\subsection{Classes of domains that satisfy the richness condition.} In this section, we show that there are non-trivial incomplete menu collections that satisfy the richness condition. A key feature of definition \ref{def: richness} is the \emph{one-step-away} property in clause (ii): a choice function that deviates at a single menu must be irrational. Thus, the notion of rationality we employ must admit violations from a single deviation. In contrast, preference maximization on arbitrary domains is characterized by the Strong Axiom of Revealed Preference (SARP), whose violations typically require multiple deviations. Hence, richness may fail.

We now introduce some preliminaries. For each $k$ with $1 \leq k \leq |X|$ let $\mathcal{A}_k = \{A \subseteq X: |A| = k\}$. Let $c$ be a choice function defined on $\mathcal{A}$. Define a revealed preference relation $R_c$ by $$
x \, R_c \, y \quad \text{if and only if} \quad \exists A \in \mathcal{A} \text{ such that } x = c(A) \text{ and } y \in A.
$$

\begin{definition}[WARP]
 A choice function $c$ satisfies the \emph{Weak Axiom of Revealed Preference (WARP)} if there are no distinct alternatives $x,y \in X$ with $xR_c y$ and $y R_c x$.
\end{definition}

\begin{lemma} \label{lemma: warp} Let $\mathcal{A}$ be a menu collection. If there is a $k$ with $3 \leq k \leq |X|$ and $\mathcal{A}_k \subseteq \mathcal{A} \subseteq \bigcup_{i=k}^{|X|} \mathcal{A}_i$ then a choice function $c$ is rational if and only if it satisfies WARP.\footnote{We note that this lemma has been generalized by \cite{caradonna2024inconsistency} with a complete characterizations of domains where WARP is necessary and sufficient.} \end{lemma}

\begin{proof}
It is immediate that $(2)\Rightarrow(1)$. We prove that $(1)\Rightarrow(2)$. To show this, it suffices to show that $R_c$ is acyclic. Suppose not, then there is a cycle $a_1R_ca_2,$ $a_2R_ca_3, \dots, a_nR_ca_1,$ where $a_1,\dots,a_n$ are distinct and $n$ is minimal. Since WARP holds it must be that $n\geq 3$.

\medskip

\noindent
\emph{Case 1 ($n \geq k$).} Consider the menu $B=\{a_1,\dots,a_k\}.$
Since $|B|=k$, we have that $B\in\mathcal A$. Let $c(B)=a_j$. If $j > 1$, then $a_{j-1}\in B$, so $a_jR_ca_{j-1}$. But by the original cycle we have $a_{j-1}R_ca_j$, hence a 2-cycle, contradicting WARP. If $j=1$, then $a_k\in B$ implies that $a_1R_ca_k$. Hence, it follows that $a_1R_ca_k, a_kR_ca_{k+1}\dots,a_nR_ca_1$ which is a cycle of length $n-k+2<n$, contradicting minimality of the original cycle.

\medskip

\noindent
\emph{Case 2 ($n<k$).}
For each $i \in \{1,\dots,n\}$ (indices modulo $n$), choose a menu $A_i\in\mathcal A$
such that $ a_i,a_{i+1}\in A_i$ and $c(A_i)=a_i$. Since $\mathcal A$ contains no menus of size smaller than $k$, each $A_i$ satisfies $|A_i|\geq k>n$. Hence for all $i$ there is an alternative $b_i$ with $b_i \in A_i \setminus  \{a_1,\dots,a_n\}$. Since $b_i\in A_i$ and $c(A_i)=a_i$, we have $ a_iR_cb_i$ for all $i \in \{1,...,n\}$. Now let $ B=\{a_1,\dots,a_n,b_1,\dots,b_{k-n}\}$.Then $|B|=k$, so $B\in\mathcal A$. There are two cases. If $c(B)=a_j$ for some $j\in\{1,\dots,n\}$, then since $a_{j-1}\in B$, we have
$a_jR_ca_{j-1}$. But we also have $a_{j-1}R_ca_j$ from the original cycle, contradicting WARP. If $c(B)=b_i$ for some $i \in \{1,...,k-n\}$ then since $a_i\in B$ we have $b_iR_ca_i$. But this contradicts $a_iR_cb_i$ and WARP.

\medskip

\noindent Case 1 and Case 2 imply that $R_c$ is acyclic. It is straightforward to show that this implies that $c$ is rational.
\end{proof} 

The next proposition shows that several natural menu collections are rich. Non-trivial examples are the menu collections $\mathcal{A}_3$ consisting of the ternary menus, or the leave-one-out menu variation $\mathcal{A}_{N-1} \cup \mathcal{A}_N$ consisting of the grand set $X$ and all subsets obtained by leaving a single alternative in $X$ out of it. 

\begin{proposition}  \label{prop: suff} Let $\mathcal{A}$ be a menu collection. If there is a $k$ with $3 \leq k \leq |X|$ and $\mathcal{A}_k \subseteq \mathcal{A} \subseteq \bigcup_{i=k}^{|X|} \mathcal{A}_i$ then $\mathcal{A}$ is a rich menu collection. \end{proposition} 

The previous proposition establishes a sufficient condition for collection $\mathcal{A}$ to be rich. This condition is, however, not necessary. There are rich menu collections $\mathcal{A}$ that do not satisfy condition $\mathcal{A}_k \subseteq \mathcal{A} \subseteq \bigcup_{i=k}^{|X|} \mathcal{A}_i$ of proposition \ref{prop: suff}. For example, one can show that the collection $\mathcal{A}_2 \cup \mathcal{A}_3$ of binary and ternary menus is rich.  

\begin{proof} Let $k$ be s.t. $\mathcal{A}_k \subseteq \mathcal{A} \subseteq \bigcup_{i=k}^{|X|} \mathcal{A}_i$. For each $P$ define a choice correspondence $C_P$ by $C_P(\{x_P(N),...,x_P(N-k)\}=\{x_P(N),...,x_P(N-k)\}$ and $C_P(A)=c_P(A)$ for all other menus $A$. In words, $C_P$ agrees with $P$ in all menus except the menu containing the worst $k$ alternatives of $P$. We claim that $C_P$ is maximally rational. 

\vspace{5mm}

\noindent Let $c$ be a choice function s.t.  $|\{A \in \mathcal{A} : c(A) \notin C_P(A)\}| = 0$. We show that $c$ is rational. To see this, note that for all menus $A \neq \{x_P(N),...,x_P(N-k)\}$ there is an alternative $x_P(i) \in A$ with $i < N-k$. Hence $c_P(A) = x_P(m)$ for some $m < N-k$ for all such menus. This, in particular, implies that if $c$ is s.t. $|\{A \in \mathcal{A} : c(A) \notin C(A)\}| = 0$ then $c(A) = c_P(A)$ for all $A \neq \{x_P(N),...,x_P(N-k)\}$. Now, it is clear that $c$ must be rational, since if not, there are $x_1,x_2$ and menus $A_1,A_2$ with $\{x_1,x_2\} \subseteq A_1,A_2$ and $x_1=c(A_1)$ and $x_2=c(A_2)$.  Then we must have $A_i = \{x_P(N),...,x_P(N-k)\}$ for some $i \in \{1,2\}$ (since, otherwise, $c_P$ would also violate WARP). Now, this implies that $x_i=c(A_i)=c(\{x_P(N),...,x_P(N-k)\})$ and $x_{j} \in \{x_P(N),...,x_P(N-k)\}$. Since $A_j \neq A_i = \{x_P(N),...,x_P(N-k)\}$ we have $x_j=c(A_j)=c_P(A_j)$, a contradiction to $c_P(A_j) \notin \{x_P(N),...,x_P(N-k)\}$. 

\vspace{5mm}

\noindent Next, we show that if $c$ is a choice function s.t.  $|\{A \in \mathcal{A} : c(A) \notin C_P(A)\}| = 1$ then $c$ is irrational. Let $A^*$ be s.t. $c(A^*) \notin C_P(A^*)$. By previous reasoning, we know that $A^* \neq \{x_P(N),...,x_P(N-k)\}$. We divide into two cases: 

\vspace{3mm}
\begin{enumerate}
\item[] Case 1: If $|A^*| > k$ then there is a menu $A$ with  $\{c(A^*),c_P(A^*)\} \subset A \subset A^*$ and $|A| = k$ and s.t. $\max(A,P) = \max(A^*, P)$. Since $c(A^*) \neq c_P(A^*)$, we have $c(A^*) R_c c_P(A^*)$ and since $c_P(A^*) = c_P(A) R_c c(A^*) $, a violation of WARP. 

\vspace{3mm}
\item[] Case 2: If $|A^*| = k$ then since $c_P(A^*) \notin \{x_P(N),...,x_P(N-k)\}$ we have $|A^* \cap \{x_P(N),...,x_P(N-k)\}| \leq k-1$. Let $b_0 \in \{x_P(N),...,x_P(N-k)\} \setminus A^*$ and let $a_0 \in A^* \setminus \{c(A^*),c_P(A^*)\} $. Set $A = A^* \setminus \{a_0\} \cup \{b_0\}$. Then $|A|=k$ so $A \in \mathcal{A}$. Further $c_P(A) = \max(A,P) = \max(A^*, P) = c_P(A^*)$. Since $c(A^*) \neq c_P(A^*)$, we have $c(A^*) R_c c_P(A^*)$ and since $c_P(A^*) = c_P(A) R_{c} c(A^*) $, a violation of WARP.  
\end{enumerate} \end{proof}

\subsection{Restating the correlation bounds.}

The partition $\{\mathcal{P}_C : C \text{ maximally rational}\}$ of $\mathcal{P}$  induces an equivalence relation on $\mathcal{P}$.  For each preference $P \in \mathcal{P}$, denote its equivalence class by $[P]$. Let $C_{[P]}$ be the unique maximally rational choice correspondence such that $ c_P(A) \in C_{[P]}(A)$ for all $ A \in \mathcal{A}. $ For each preference $P$, define $$ \mathbb{C}^{\rho}_{[P]} = \sum_{A \in \mathcal{A}} \rho(C_{[P]}(A),A).$$

\begin{definition}[Correlation Bounds] \label{def: bounds}
A stochastic choice function $\rho$ satisfies the \emph{Correlation Bounds} if for all preferences $P$ it holds that  
$$\mathbb{C}^{\rho}_{[P]} - [K - 1] \leq 0.$$
\end{definition}

\begin{definition}[Measure of correlation between preference and choice functions] Fix a preference $P$, then for all choice functions $c$:
$$n([P],c) = |\{A \in \mathcal{A} :c(A) \in  C_{[P]}(A)  \}|.$$    
\end{definition}

The measure $n([P],c)$ captures the correlation between a preference $P$ and a choice function by counting their overlapping choices.

\begin{lemma} 
Let $\rho$ be a stochastic choice function. For all $\mu$ s.t. $\rho=\rho_{\mu}$, and for all preferences $P$, it holds that
$$\sum_{A \in \mathcal{A}} \rho(C_{[P]}(A),A) = \sum\limits_{c \in \mathcal{C}} \mu(c) n([P],c).$$
\end{lemma}
\begin{proof} The proof follows by noting that $$ \sum_{A \in \mathcal{A}} \rho(C_{[P]}(A),A) = \sum_{A \in \mathcal{A} }\sum_{c \in \mathcal{C}}[\mu(c)\bold{1}\{c(A) \in C_{[P]}(A) \} ]=$$
$$\sum_{c \in \mathcal{C}} \mu(c) \sum_{A \in \mathcal{A}} \bold{1}\{c(A) \in C_{[P]}(A) \} =  \sum_{c \in \mathcal{C}} \mu(c)n([P], c).$$
The first equality follows since $\rho=\rho_{\mu}$, the second equality follows by changing the order of summation, and the final one follows by definition of $n([P], c)$. \end{proof}

\subsection{Main Results}

We now state our main results. Theorem \ref{thm: RNUM-gen} characterizes the set of stochastic choice functions with an irrational representation and Theorem \ref{thm: main-gen} characterizes the set of I-RUMs.

\begin{theorem} \label{thm: RNUM-gen}
If $\mathcal{A}$ is a rich menu collection, a stochastic choice function $\rho$ is an RNUM if and only if it satisfies the correlation bounds.
\end{theorem}

\begin{theorem} \label{thm: main-gen}
Let $\rho$ be a RUM. If $\mathcal{A}$ is a rich menu collection, $\rho$ is an I-RUM if and only if it satisfies the correlation bounds.
\end{theorem}

\subsection{Proof of Theorem \ref{thm: RNUM-gen}}  \label{app: main-gen}

\subsubsection{The necessity of the correlation bounds} \label{sec: nproof-gen}
The necessity of the correlation bounds follows from an application of the Fréchet bounds. Recall that $| \mathcal{A} | = K$.

\begin{lemma} \label{lemma: strict-gen} Let $\rho$ be a stochastic choice function and $P$ be a preference. It then holds that $$ \max\{0,  \sum_{A \in \mathcal{A}} \rho(C_{[P]}(A),A) - [K-1]\} \leq  \mu( [P] )$$ for all RCMs $\mu$ such that $\rho= \rho_{\mu}$. \end{lemma}

\begin{proof} Let $\rho$ be a stochastic choice function and $P$ be a preference. Note that $\rho(C_{[P]}(A),A) =\mu(\{c \in \mathcal{C}: c(A) \in C_{[P]}(A)\})$. Thus, it follows by the Fréchet inequality that $$ \max\{0,  \sum_{A \in \mathcal{A}} \rho(C_{[P]}(A),A) - [K-1]\} \leq \mu(\cap_{A \in \mathcal{A}} \{c \in \mathcal{C}: c(A) \in C_{[P]}(A)\}) = \mu([P]),$$ where the final equality follows by construction of the equivalence classes $[P]$.  \end{proof}

Since any RNUM $\mu$ puts zero probability on all the rational choice functions, i.e. $\mu(c_P)=0$ for all preferences $P$, Lemma \ref{lemma: strict-gen} implies that any RNUM must satisfy the following inequalities for all preferences $P$:  
$$\mathbb{C}^{\rho}_{[P]} - [K-1] = \sum_{A \in \mathcal{A}} \rho(C_{[P]}(A),A) - [K-1]  \leq  0,$$
which are exactly the correlation bounds. 


\subsubsection{The sufficiency of the correlation bounds} \label{app: sufficiency}

We refer the reader to the main text for an exposition and a sketch of the general idea behind the proof. Let $\rho$ be a RUM with distribution $\mu$. We say that $\rho$ \emph{satisfies the correlation bounds with strict inequality} if $\mathbb{C}^{\rho}_{[P]} - [K-1] < 0$ for all preferences $P$. Similarly we say that $\rho$ \emph{satisfies the correlation bounds with equality} if $\mathbb{C}^{\rho}_{[P]} - [K-1]  = 0$ for some preference $P$ and $\mathbb{C}^{\rho}_{[P']} -[K-1] \leq 0$ for all preferences $P'$. We will sometimes also say that RCM $\mu$ satisfies the correlation bounds (with equality) if $\rho_{\mu}$ satisfies the correlation bounds (with equality).

\begin{lemma} The subset of SCFs that satisfy the correlation bounds is a convex set. I.e. the set of all SCFs $\rho$ such that $$\mathbb{C}^{\rho}_{[P]} - [K-1] \leq 0$$ for all preferences $P$ is convex. \end{lemma}

\begin{proof} 
The proof is straightforward and follows by expanding the definition of correlation bounds.
\end{proof}

\begin{lemma} Let $\rho_{\mu}$ be a Dual RCM such that $\mu(P)> 0$ and $\mathbb{C}^{\rho}_{[P]} - [K-1]=0$ for some preference $P$, then $\rho$ is an RNUM. \end{lemma}

\begin{proof} 
Let $\rho_{\mu}$ be a Dual RCM such that $\mu(P)> 0$ and $\mathbb{C}^{\rho}_{[P]} - [K-1]=0$ for some preference $P$. Thus, $\rho$ has support $\{c,P\}$ where $c \in \mathcal{C} \setminus [P]$. Since $\mathbb{C}^{\rho}_{[P]}-[K-1]=0$ it follows that $\mu(c)n([P],c) + \mu(P)n([P],P)=K-1$. Rearranging the latter equality it follows that $\mu(P) = \frac{K- 1 - n([P],c) }{K-n([P],c)}$. Let $k \in \mathbb{N}$ be  such that $k \mu(c) = \mu(P)$. I.e. $k= K- 1- n([P],c) $. Let the set $\mathcal{D}(P)$ be defined as $\mathcal{D}(P)= \{A\in \mathcal{A} : c(A) \notin C_{[P]}(A)\}$ (which has cardinality $k+1$). 
Pick an arbitrary $B_0 \in \mathcal{D}(P)$ and define a choice function $c_{B_0}$ by 
$$ c_{B_0}(A) =
\begin{cases} 
c(A) & \text{if $A \in \{B_0\} \cup \mathcal{A} \setminus \mathcal{D}(P)  $} \\
c_{P}(A) & \text{otherwise.} \\
\end{cases}
$$

For each $B \in \mathcal{D}(P) \setminus \{B_0\}$ define a choice function $c_B$ by 
$$ c_{B}(A) =
\begin{cases} 
c(A) & \text{if $A = B$} \\
c_{P}(A) & \text{if $A  \neq  B$} \\
\end{cases}
$$
\vspace{5mm}

\noindent Note that each $c_B$ is irrational (since $|\{A \in \mathcal{A}: c_B(A) \notin C_{[P]}(A) \}|=1$). Further, $c_{B_0}$ is irrational since $c(A) \in C_{[P]}(A)$ for all $A \in \mathcal{A} \setminus \mathcal{D}(P)$ and hence $|\{A \in \mathcal{A}: c_{B_0}(A) \notin C_{[P]}(A) \}|=|\{B_0\}|=1$. Let $\mu'$ be the uniform distribution on $\{c_B : B \in \mathcal{D}(P)\}$, i.e.$$\mu'(c_B) =\frac{1}{k+1}$$ for all $B \in \mathcal{D}(P)$. To show that $\mu$ is an RNUM, it remains to show that the stochastic choices generated by $\mu$ and $\mu'$ are the same. There are two cases.

\vspace{5mm}

\noindent Case 1. Let $A \notin \mathcal{D}(P)$, then the stochastic choices generated by $\mu$ and $\mu'$ are the same (i.e. that $\rho_{\mu}(\cdot,A)=\rho_{\mu'}(\cdot,A)$).
\vspace{5mm}

\noindent Case 2. If $A= B$ for some $B \in \mathcal{D}(P)$. Then $c_B(A) = c(A)$ and $c_{B'}(A) = c_{P}(A)$ for all $B' \in \mathcal{D}(P) \setminus \{B\}$. It hence follows that 
 
 $$\sum_{B' \in \mathcal{D}(P)} \mu'(c_{B'}) \bold{1}\{a =c_{B'}(A)\} =$$  $$=\frac{1}{k+1}\bold{1}\{a =c_B(A)\}  +  \sum_{B' \in \mathcal{D}(P) \setminus \{B\}} \frac{1}{k+1} \bold{1}\{a =c_{B'}(A)\}  = $$   $$=\frac{1}{k+1} \bold{1}\{a =c_B(A)\} + \frac{k}{k+1} \bold{1}\{a =c_{P}(A)\} $$ 
 
 $$ =\mu(c)\bold{1}\{a =c(A)\} + \mu(P) \bold{1}\{a =c_{P}(A)\},$$ where the last equality follows since $k\mu(c) = \mu(P)$ implies that $\mu(c)=\frac{1}{k+1}$ and $\mu(P)=\frac{k}{k+1}$. \end{proof}

\begin{lemma} \label{lemma: correq-gen}
Let $\mu$ be an RCM such that the correlation bounds are satisfied with equality. Then there are Dual RCMs $(\hat{\mu_i})_{i=1}^{n}$ such that 

\begin{enumerate}

\item \label{lemma: eq1-gen} $\mu$ is a convex combination of the $(\hat{\mu_i})_{i=1}^{n}$, i.e. there are weights $\delta_i$ s.t. $\mu = \sum_{k=1}^n \delta_i \hat{\mu_i}$,

\item \label{lemma: eq2-gen} each Dual RCM $\hat{\mu_i}$ satisfies the correlation bounds with equality.  

\end{enumerate}

\end{lemma}

In the following proof, we will make use of a special class of RCMs that we call \emph{almost dual RCMs}. An almost dual RCM is an RCM $\mu$ such that $\text{supp}\, {\mu} \subseteq [P] \cup \{c\}$ for some preference $P \in \mathcal{P}$ and choice function $c \in \mathcal{C}$. The proof is split into two steps. STEP 1 shows that $\mu$ is a convex combination of almost dual RCMs $\mu_i$ where each almost dual RCM satisfies the bounds with equality. STEP 2 then shows that every almost dual RCM that satisfies the correlation bounds with equality is a convex combination of dual RCMs that satisfy the correlation bounds with equality. 

\begin{proof} Let $\mu$ be an RCM that satisfies the correlation bounds with equality. W.l.o.g. assume that the correlation bound is satisfied with equality at $P$ (i.e. we have $ K - 1 = \sum\limits_{c \in \mathcal{C}} n([P], c) \mu(c)$). 


\vspace{5mm}

\noindent STEP 1: We next construct a series of almost dual RCMs $(\hat{\mu}_c)_{c \in \mathcal{C}\setminus[P]}$  such that $\mu$ is a convex combination of $(\hat{\mu}_c)_{c \in \mathcal{C}\setminus[P]}$ and such that each $\hat{\mu}_c$ satisfies the correlation bounds with equality. 

\vspace{5mm}

 \noindent For each $c \in \mathcal{C} \setminus [P]$ we define an almost dual RCM $\hat{\mu}_c$ with support on $[P] \cup \{c\}$ such that the correlation bounds are satisfied with equality. For all $Q \in [P]$ set 

$$ \hat{\mu}_c(Q) =\frac{\mu(Q)}{\mu([P])} \left[ \frac{K - 1 - n([P], c)}{K - n([P],c)} \right]$$


\noindent and set 
$$\hat{\mu}_c(c) = \frac{1}{K - n([P],c)}.$$ 

\vspace{5mm}

\noindent Note that the probabilities above are well-defined since $n([P],c) < K$ for all $c \in \mathcal{C} \setminus [P]$. Further, it is easily verified that they are positive and sum to one (and hence probability measures). 

\vspace{5mm}

\noindent A sequence of weights $(\delta_c)_{c \in \mathcal{C} \setminus [P]}$  is then defined for each $c \in \mathcal{C} \setminus [P]$ as follows:
 $$ \delta_c =  \mu(c) [K - n([P], c)].$$
 \noindent We next show that $\sum_{c \in \mathcal{C} \setminus [P]} \delta_c =1$, i.e. that the $\delta_c$ are indeed weights (that $\delta_c \geq 0$ is obvious). We have that 
$$\sum_{c \in \mathcal{C} \setminus [P]} \delta_c  = \sum_{c \in \mathcal{C} \setminus [P]} \mu(c)[K - n([P], c)]  = $$
$$K [1- \mu([P])] - \sum_{c \in \mathcal{C} \setminus [P]}  n([P],c) \mu(c) =$$
$$1 + K-1 - \sum_{c \in \mathcal{C} } n([P], c) \mu(c) = 1,$$ where the first equality follows by definition of  $\delta_c$, the second equality follows since $\sum_{c \in \mathcal{C} \setminus [P]} K \mu(c)=K[1- \mu([P])]$, the third equality follows by noting that $\sum_{Q \in [P]}n([P], Q)\mu(Q) =K \mu([P])$ and rearranging, and the final equality follows since the correlation bound holds with equality at $P$. 

\vspace{5mm} 

\noindent It remains to show that $\mu$ is a convex combination of the almost Dual RCMs  $(\hat{\mu}_c)_{ c \in \mathcal{C}\setminus  [P]}$. It is clear that for all $c \in \mathcal{C} \setminus [P]$ it holds that
$$\mu(c) = \delta_c \hat{\mu}_c(c) =\sum_{c' \in \mathcal{C} \setminus [P]} \delta_{c'}\hat{\mu}_{c'}(c),$$ where the first equality follows by construction of $\delta_c$ and $\hat{\mu}_c$, and the second equality follows since $\hat{\mu}_{c'}(c)= 0$ for all $c' \notin [P] \cup \{c\}$. We next note that for all $Q \in [P]$
$$\sum_{c \in \mathcal{C} \setminus [P]} \delta_c \hat{\mu}_c(Q)=$$
$$ \frac{\mu(Q)}{\mu([P])} \sum_{c \in \mathcal{C} \setminus [P]} \mu(c) [K -1- n([P], c)] = $$
$$ \frac{\mu(Q)}{\mu([P])} \left[\sum_{c \in \mathcal{C} \setminus [P]} \mu(c) [K - n([P], c)] + \mu([P])-1 \right]= $$ 
$$\frac{\mu(Q)}{\mu([P])} \mu([P])= \mu(Q).$$ 

\vspace{5mm}

\noindent STEP 2: We next show that if $\mu$ is an almost dual RCM that satisfies the correlation bounds with equality then it is a convex combination of dual RCMs that satisfy the correlation bounds with equality. Suppose that $\mu$ satisfies the correlation bound with equality at $P$ and that $\mu$ has support $[P] \cup \{c\}$. We may then define for each $Q \in [P]$ a dual RCM $\mu_Q$ with support $\{Q,c\}$ by 
$$\mu_Q(Q) = \frac{K - 1 - n([P], c)}{K - n([P],c)}$$
$$\mu_Q(c) = \frac{1}{K - n([P],c)}.$$

\noindent By construction $\mu_Q$ satisfies the correlation bounds for all $Q \in [P]$. Further, define weights $\delta_Q = \frac{\mu(Q)}{\mu([P])}  \mu(c) [K - n([P], c)] $. It is straightforward to verify that $(\delta_Q)_{Q \in [P]}$ are indeed weights and that $\mu =\sum_{Q \in [P]} \delta_Q \mu_Q$, which completes the proof. \end{proof}

\begin{lemma}  \label{lemma: extreme-gen} Let $\rho$ be an RCM with distribution $\mu$ (so $\rho=\rho_{\mu}$) that is an extreme point of the set of stochastic choice functions that satisfy the correlation bounds. Then $\rho$ is either an RNUM or $\rho$ satisfies the Correlation bounds with equality.  \end{lemma}

\begin{proof} If $\rho$ is an RNUM the proof is complete. Thus, assume that $\rho$ is an RCM with distribution $\mu$ with support on at least one preference $P_0$. We show that if $\rho=\rho_{\mu}$ satisfies the correlation bounds with strict inequality (i.e. $\mathbb{C}^{\rho}_{[P]} - [K-1]<0$ for all preferences $P$) then $\rho$ is not an extreme point. By assumption, the support of $\mu$ is not contained in $[P_0]$ (otherwise $\mathbb{C}^{\rho}_{[P_0]}=K$). Then there is a choice function $c_1 \in \mathcal{C} \setminus [P_0]$ such that $\{c_{P_0},c_1\}$ is contained in the support of $\mu$. Let $\alpha$ be a distribution with support on $c_{P_0}, c_1$ such that the Dual RCM $\rho^d$, defined for all $a \in A$ and $A \in \mathcal{A}$ by $$\rho^d(a,A) = \alpha(c_{P_0}) \bold{1}\{a= c_{P_0}(A)\} + \alpha(c_1) \bold{1}\{a= c_1(A)\},$$ satisfies the correlation bounds with equality at $P_0$ (hence at all preferences, as $\alpha(c_1)\leq\frac{1}{2}$ unless $\rho^d$ is concentrated on the irrational $c_1$). For each $\beta \in [0,1)$ such that $\mu(c) - \beta \alpha(c) > 0$ for all $c \in \{c_{P_0},c_1\}$  (or equivalently $\beta < \min\{\frac{\mu(c_{P_0})}{\alpha(c_{P_0})},  \frac{\mu(c_1)}{\alpha(c_1)}\}$) define an RCM $\mu^*_{\beta}$  by  
 $$\mu^*_{\beta}(c) =
  \begin{cases} 
  \frac{\mu(c) - \beta \alpha(c) }{1-\beta} & \text{if  $c \in \{c_{P_0},c_1\}$} \\ 
 \frac{\mu(c)}{1 - \beta} & \text{otherwise.}
  \end{cases}
  $$ 
  \vspace{5mm}
\noindent Let $\rho^*(\beta) = \rho_{\mu^*_{\beta}}$, i.e. $\rho^*(\beta)$ is the RCM stochastic choice function with distribution $\mu^*_{\beta}$. For each preference $P$ define a function $f_P$ by $$f_P(\beta) = \mathbb{C}^{\rho^*(\beta)}_{[P]}$$ for all $\beta \in \left[0, \min \left\{\frac{\mu(c_{P_0})}{\alpha(c_{P_0})},  \frac{\mu(c_1)}{\alpha(c_1)}\right\} \right)$. First note that $f_P(0)=\mathbb{C}^{\rho^*(0)}_{[P]}=\mathbb{C}^{\rho}_{[P]}$ for all preferences $P$. Hence, since the correlation bounds are satisfied with strict inequality we have $f_P(0) < K-1$ for all preferences $P$. Define a (vector-valued) function $f(\beta): =(f_P(\beta))_{P \in \mathcal{P}}$. Since each $f_P$ is a continuous function of $\beta$ it follows that $f$ is a continuous function of $\beta$ and hence there is a $\hat{\beta} > 0$ such that $f_P(\beta) < K-1$ for all $0 \leq \beta < \hat{\beta}$ and for all preferences $P$. Let $\beta$  be such that $0<\beta < \hat{\beta} <\min \left\{\frac{\mu(c_{P_0})}{\alpha(c_{P_0})},  \frac{\mu(c_1)}{\alpha(c_1)}\right\}$. We then have that $\rho = \beta \rho^d + (1-\beta) \rho^*(\beta)$ and since both $\rho^d$ and $\rho^*(\beta)$ satisfy the correlation bounds, it follows that $\rho$ is not an extreme point (as it can be written as a non-trivial convex combination of two other RCMs that satisfy the correlation bounds).  \end{proof}

\bibliographystyle{apa-good}
{\small
\bibliography{bibliography}

@article{petri23,
  title={Binary single-crossing random utility models},
  author={Petri, Henrik},
  journal={Games and Economic Behavior},
  volume={138},
  pages={311--320},
  year={2023},
  publisher={Elsevier}
}

@book{Rock70,
url = {https://doi.org/10.1515/9781400873173},
title = {Convex Analysis},
author = {Ralph Tyrell Rockafellar},
publisher = {Princeton University Press},
address = {Princeton},
doi = {doi:10.1515/9781400873173},
isbn = {9781400873173},
year = {1970},
lastchecked = {2023-12-21}
}

@article{apesteguia2017single,
  title={Single-crossing random utility models},
  author={Apesteguia, Jose and Ballester, Miguel A and Lu, Jay},
  journal={Econometrica},
  volume={85},
  number={2},
  pages={661--674},
  year={2017},
  publisher={Wiley Online Library}
}

@article{kitamura2018nonparametric,
  title={Nonparametric analysis of random utility models},
  author={Kitamura, Yuichi and Stoye, J{\"o}rg},
  journal={Econometrica},
  volume={86},
  number={6},
  pages={1883--1909},
  year={2018},
  publisher={Wiley Online Library}
}

@article{deb2023revealed,
  title={Revealed price preference: theory and empirical analysis},
  author={Deb, Rahul and Kitamura, Yuichi and Quah, John KH and Stoye, J{\"o}rg},
  journal={The Review of Economic Studies},
  volume={90},
  number={2},
  pages={707--743},
  year={2023},
  publisher={Oxford University Press US}
}

@article{smeulders2021nonparametric,
  title={Nonparametric analysis of random utility models: computational tools for statistical testing},
  author={Smeulders, Bart and Cherchye, Laurens and De Rock, Bram},
  journal={Econometrica},
  volume={89},
  number={1},
  pages={437--455},
  year={2021},
  publisher={Wiley Online Library}
}

@article{im2022non,
  title={Non-rationalizable individuals and stochastic rationalizability},
  author={Im, Changkuk and Rehbeck, John},
  journal={Economics Letters},
  volume={219},
  pages={110786},
  year={2022},
  publisher={Elsevier}
}

@article{mcfadden1990stochastic,
  title={Stochastic rationality and revealed stochastic preference},
  author={McFadden, Daniel and Richter, Marcel K},
  journal={Preferences, Uncertainty, and Optimality, Essays in Honor of Leo Hurwicz, Westview Press: Boulder, CO},
  pages={161--186},
  year={1990}
}

@article{kashaev2022random,
  title={A random attention and utility model},
  author={Kashaev, Nail and Aguiar, Victor H},
  journal={Journal of Economic Theory},
  volume={204},
  pages={105487},
  year={2022},
  publisher={Elsevier}
}

@inproceedings{mcfadden2006revealed,
  title={Revealed stochastic preference: a synthesis},
  author={McFadden, Daniel L},
  booktitle={Rationality and Equilibrium: A Symposium in Honor of Marcel K. Richter},
  pages={1--20},
  year={2006},
  organization={Springer}
}

@article{giarlotta2023context,
  title={Context-sensitive rationality: Choice by salience},
  author={Giarlotta, Alfio and Petralia, Angelo and Watson, Stephen},
  journal={Journal of Mathematical Economics},
  volume={109},
  pages={102913},
  year={2023},
  publisher={Elsevier}
}

@article{tversky1972choice,
  title={Choice by elimination},
  author={Tversky, Amos},
  journal={Journal of mathematical psychology},
  volume={9},
  number={4},
  pages={341--367},
  year={1972},
  publisher={Elsevier}
}

@article{cerreia2021canon,
  title={A canon of probabilistic rationality},
  author={Cerreia-Vioglio, Simone and Lindberg, Per Olov and Maccheroni, Fabio and Marinacci, Massimo and Rustichini, Aldo},
  journal={Journal of Economic Theory},
  volume={196},
  pages={105289},
  year={2021},
  publisher={Elsevier}
}

@article{giarlotta2022bounded,
  title={Bounded rationality is rare},
  author={Giarlotta, Alfio and Petralia, Angelo and Watson, Stephen},
  journal={Journal of Economic Theory},
  volume={204},
  pages={105509},
  year={2022},
  publisher={Elsevier}
}

@article{apesteguia2021separating,
  title={Separating predicted randomness from residual behavior},
  author={Apesteguia, Jose and Ballester, Miguel A},
  journal={Journal of the European Economic Association},
  volume={19},
  number={2},
  pages={1041--1076},
  year={2021},
  publisher={Oxford University Press}
}

@Article{luce59,
author = {Duncan R. Luce},
title = {Individual choice behavior: a theoretical analysis},
journal = {},
year = {1959},
key = {•},
volume = {},
number = {},
pages = {},
month = {Wiley, New York}
}

@article{stoye2019revealed,
  title={Revealed Stochastic Preference: A one-paragraph proof and generalization},
  author={Stoye, J{\"o}rg},
  journal={Economics Letters},
  volume={177},
  pages={66--68},
  year={2019},
  publisher={Elsevier}
}

@article{arrow1959rational,
  title={Rational choice functions and orderings},
  author={Arrow, Kenneth J},
  journal={Economica},
  volume={26},
  number={102},
  pages={121--127},
  year={1959}
}

@Article{mariottidual,
author={Paola Manzini and Marco Mariotti},
title={Dual random utility maximisation},
journal={Journal of Economic Theory},
year={2018},
volume={177},
pages={ 162-182},
month={}
}

@article{masa2012,
  title={Revealed attention},
  author={Masatlioglu, Yusufcan and Nakajima, Daisuke and Ozbay, Erkut Y},
  journal={American Economic Review},
  volume={102},
  number={5},
  pages={2183--2205},
  year={2012}
}

@article{barbera1986falmagne,
  title={Falmagne and the rationalizability of stochastic choices in terms of random orderings},
  author={Barber{\'a}, Salvador and Pattanaik, Prasanta K},
  journal={Econometrica},
  volume={54},
  number={3},  
  pages={707--715},
  year={1986},
  publisher={}
}

@article{falmagne1978representation,
  title={A representation theorem for finite random scale systems},
  author={Falmagne, Jean-Claude},
  journal={Journal of Mathematical Psychology},
  volume={18},
  number={1},
  pages={52--72},
  year={1978},
  publisher={Elsevier}
}

@article{fiorini2004short,
  title={A short proof of a theorem of {F}almagne},
  author={Fiorini, Samuel},
  journal={Journal of Mathematical Psychology},
  volume={48},
  number={1},
  pages={80--82},
  year={2004},
  publisher={Elsevier}
}

@article{mccausland2020testing,
  title={Testing the random utility hypothesis directly},
  author={McCausland, William J and Davis-Stober, Clintin and Marley, Anthony AJ and Park, Sanghyuk and Brown, Nicholas},
  journal={The Economic Journal},
  volume={130},
  number={625},
  pages={183--207},
  year={2020},
  publisher={Oxford University Press}
}

@article{filiz2023progressive,
  title={Progressive random choice},
  author={Filiz-Ozbay, Emel and Masatlioglu, Yusufcan},
  journal={Journal of Political Economy},
  volume={131},
  number={3},
  pages={716--750},
  year={2023},
  publisher={The University of Chicago Press Chicago, IL}
}

@Article{DMMT,
author = {Valentino Dardanoni and Paola Manzini and Marco Mariotti and Christopher Tyson},
title = {Inferring cognitive heterogeneity from aggregate choices},
journal = {Econometrica},
year = {2020},
key = {•},
volume = {88},
number = {3},
pages = { 1269-1296},
month = {May}
}

@article{gilboa1990necessary,
  title={A necessary but insufficient condition for the stochastic binary choice problem},
  author={Gilboa, Itzhak},
  journal={Journal of Mathematical Psychology},
  volume={34},
  number={4},
  pages={371--392},
  year={1990},
  publisher={Elsevier}
}

@article{hoderlein2015testing,
  title={Testing stochastic rationality and predicting stochastic demand: the case of two goods},
  author={Hoderlein, Stefan and Stoye, J{\"o}rg},
  journal={Economic Theory Bulletin},
  volume={3},
  number={2},
  pages={313--328},
  year={2015},
  publisher={Springer}
}

@article{matzkin2007heterogeneous,
  title={Heterogeneous choice},
  author={Matzkin, Rosa L},
  journal={Econometric Society Monographs},
  volume={43},
  pages={75},
  year={2007},
  publisher={Cambridge University Press}
}

@article{hoderlein2014revealed,
  title={Revealed preferences in a heterogeneous population},
  author={Hoderlein, Stefan and Stoye, J{\"o}rg},
  journal={Review of Economics and Statistics},
  volume={96},
  number={2},
  pages={197--213},
  year={2014},
  publisher={The MIT Press}
}

@article{blundell2003nonparametric,
  title={Nonparametric {E}ngel curves and revealed preference},
  author={Blundell, Richard W and Browning, Martin and Crawford, Ian A},
  journal={Econometrica},
  volume={71},
  number={1},
  pages={205--240},
  year={2003},
  publisher={Wiley Online Library}
}

@article{grandmont1992transformations,
  title={Transformations of the commodity space, behavioral heterogeneity, and the aggregation problem},
  author={Grandmont, Jean-Michel},
  journal={Journal of Economic Theory},
  volume={57},
  number={1},
  pages={1--35},
  year={1992},
  publisher={Elsevier}
}

@article{monderer1992stochastic,
  title={The stochastic choice problem: A game-theoretic approach},
  author={Monderer, Dov},
  journal={Journal of Mathematical Psychology},
  volume={36},
  number={4},
  pages={547--554},
  year={1992},
  publisher={Elsevier}
}

@article{aguiar2021stochastic,
  title={Stochastic revealed preferences with measurement error},
  author={Aguiar, Victor H and Kashaev, Nail},
  journal={The Review of Economic Studies},
  volume={88},
  number={4},
  pages={2042--2093},
  year={2021},
  publisher={Oxford University Press}
}

@article{chambers2025ordered,
author = {Christopher P. Chambers and Yusufcan Masatlioglu and Kemal Yildiz},
title = {Ordered Probabilistic Choice},
  journal={arXiv preprint: 2504.01227v2},
  year={2025}
}

@techreport{caradonna2024inconsistency,
  title={The Inconsistency Rank},
  author={Caradonna, Peter},
  year={2026},
  institution={Mimeo}
}

@article{dardanoni2023mixture,
  title={Mixture choice data: revealing preferences and cognition},
  author={Dardanoni, Valentino and Manzini, Paola and Mariotti, Marco and Petri, Henrik and Tyson, Christopher J},
  journal={Journal of Political Economy},
  volume={131},
  number={3},
  pages={687--715},
  year={2023},
  publisher={The University of Chicago Press Chicago, IL}
}

@article{suleymanov2024branching,
  title={Branching-independent random utility model},
  author={Suleymanov, Elchin},
  journal={Journal of Economic Theory},
  pages={105880},
  year={2024},
  publisher={Elsevier}
}

@book{luce1957introduction,
  title={Games and decisions: Introduction and critical survey},
  author={Luce, R Duncan and Raiffa, Howard},
  year={1957},
  publisher={}
}

@article{chambers2025limits,
  title={The limits of identification in discrete choice},
  author={Chambers, Christopher P and Turansick, Christopher},
  journal={Games and Economic Behavior},
  volume={150},
  pages={537--551},
  year={2025},
  publisher={Elsevier}
}

@Article{turansick,
author = {Christopher Turansick},
title = {Identification in the random utility model},
journal = {Journal of Economic Theory},
year = {2022},
key = {•},
volume = {203},
number = {},
pages = {},
month = {},
}

@inproceedings{fishburn1998,
  title={Stochastic Utility},
  author={Peter C. Fishburn},
  booktitle={Handbook of Utility Theory},
  pages={273-380},
  year={1998},
  organization={Kluwer Dordrecht}
}

@article{cattaneo2020random,
  title={A random attention model},
  author={Cattaneo, Matias D and Ma, Xinwei and Masatlioglu, Yusufcan and Suleymanov, Elchin},
  journal={Journal of Political Economy},
  volume={128},
  number={7},
  pages={2796--2836},
  year={2020},
  publisher={The University of Chicago Press Chicago, IL}
}

@Article{sen71,
author = {Amartya Sen},
title = {Choice Functions and Revealed Preference},
journal = {The Review of Economic Studies},
year = {1971},
key = {•},
volume = {38},
number = {3},
pages = {  307-317},
month = {Jul.},
}

@Article{block,
author = {H.D. Block and Jacob Marschak},
title = {Random Orderings and Stochastic Theories of Responses},
journal = {Contributions to Probability and Statistics},
year = {1960},
key = {},
volume = {},
number = {Stanford University Press},
pages = { },
month = {},
}

@Article{fudenberg,
author = {Drew Fudenberg and Ryota Iijima and Tomasz Strzalecki},
title = {Stochastic choice and revealed perturbed utility},
journal = {Econometrica},
year = {2015},
key = {},
volume = {83},
number = {6},
pages = { 2371-2409},
month = {Nov.},
}

@Article{kalai,
author = {Gil Kalai and Ariel Rubinstein and Ran Spiegler},
title = {Rationalizing choice functions by multiple rationales},
journal = {Econometrica},
year = {2002},
key = {•},
volume = {70},
number = {6},
pages = { 2481-2488},
month = {Nov.},
}

@Article{becker,
author = {Gary S. Becker},
title = {Irrational behavior and economic theory},
journal = {Journal of Political Economy},
year = {1962},
key = {•},
volume = {},
number = {70},
pages = { 1-13},
month = {},
}
}

\end{document}